%% file: main.tex
\documentclass[conference]{IEEEtran}
\makeatletter
\def\ps@headings{%
\def\@oddhead{\mbox{}\scriptsize\rightmark \hfil \thepage}%
\def\@evenhead{\scriptsize\thepage \hfil \leftmark\mbox{}}%
\def\@oddfoot{}%
\def\@evenfoot{}}
\makeatother
\usepackage{array, calc, color, multicol, multirow}

\usepackage{mathtools}

\usepackage{textcomp, mathrsfs, stmaryrd, setspace}
\usepackage{amssymb}
\usepackage{soul}
\usepackage[dvipsnames]{xcolor}
\usepackage{cite}
\usepackage{url}
\usepackage{amsmath,stackengine}
\usepackage{amsthm}
\stackMath
\usepackage{algorithm}
\usepackage{algpseudocode}
\usepackage{caption}
\usepackage{booktabs}

\usepackage{tikz}
\usetikzlibrary{shapes.geometric, arrows, positioning, arrows.meta}
\usepackage{graphicx}

\usepackage{algpseudocode}
\algnewcommand{\algorithmicinput}{\textbf{Input:}}
\algnewcommand{\algorithmicoutput}{\textbf{Output:}}
\algnewcommand{\Input}{\item[\algorithmicinput]}
\algnewcommand{\Output}{\item[\algorithmicoutput]}

\usepackage{makecell}
\usepackage{pifont}
\usepackage{colortbl}
\usepackage{xcolor}
\usepackage{tikz}
\usepackage{wrapfig}
\usepackage{booktabs}
\usepackage{array}

\newcommand{\Prob}{\mathbb{P}}

\newtheorem{theorem}{Theorem}

\newtheorem{definition}{Definition}

\usepackage{xspace,exscale,relsize}

\usepackage{graphics, epstopdf, graphicx, epsfig, psfrag, epsf,subfigure}

\newcommand{\pmdi}{{\tt{privateMDI}}\xspace}

\usepackage[left=0.625in,right=0.625in,top=0.75in]{geometry}

\IEEEoverridecommandlockouts

\begin{document}


\title{Privacy-Preserving Hierarchical Model-Distributed Inference}

\author{\IEEEauthorblockN{Fatemeh Jafarian Dehkordi}
\IEEEauthorblockA{
\textit{University of Illinois Chicago}\\
fjafar3@uic.edu}
\and
\IEEEauthorblockN{ Yasaman Keshtkarjahromi}
\IEEEauthorblockA{
\textit{Seagate Technology}\\
yasaman.keshtkarjahromi@seagate.com}
\and
\IEEEauthorblockN{Hulya Seferoglu}
\IEEEauthorblockA{
\textit{University of Illinois Chicago}\\
hulya@uic.edu}
\thanks{This work was supported in parts by the Army Research Office (W911NF2410049), the National Science Foundation (CCF-1942878, CNS-2148182, CNS-2112471, CNS-1801708), and Seagate Technology (00118496.0).}
}

\maketitle

{$\hphantom{a}$}{}

\input{abstract}

\input{introduction}

\input{related_work}
\input{system_model}

\input{overall_offline_protocol}

\input{overall_online_protocol}

\input{analysis}

\input{evaluation}
\input{conclusion}

\bibliographystyle{IEEEtran}

\bibliography{refs}

\input{proxy_ot_proof}

\input{protocol_proof}

\input{analysis_extended}

\input{evaluation_extended}

\end{document}

%% file: abstract.tex
\begin{abstract}

This paper focuses on designing a privacy-preserving Machine Learning (ML) inference protocol for a hierarchical setup, where clients own/generate data, model owners (cloud servers) have a pre-trained ML model, and edge servers perform ML inference on clients' data using the cloud server's ML model. Our goal is to speed up ML inference while providing privacy to both data and the ML model. Our approach (i) uses model-distributed inference (model parallelization) at the edge servers and (ii) reduces the amount of communication to/from the cloud server. Our privacy-preserving hierarchical model-distributed inference, \pmdi design uses additive secret sharing and linearly homomorphic encryption to handle linear calculations in the ML inference, and garbled circuit and a novel three-party oblivious transfer are used to handle non-linear functions. \pmdi consists of offline and online phases. We designed these phases in a way that most of the data exchange is done in the offline phase while the communication overhead of the online phase is reduced. In particular, there is no communication to/from the cloud server in the online phase, and the amount of communication between the client and edge servers is minimized. The experimental results demonstrate that \pmdi significantly reduces the ML inference time as compared to the baselines.




\end{abstract}

%% file: introduction.tex
\section{Introduction} \label{sec}

Machine learning (ML) has become a powerful tool for supporting applications such as mobile healthcare, self-driving cars, finance, marketing, agriculture, etc. These applications generate vast amounts of data at the edge, requiring swift processing for timely responses. On the other hand, ML models are getting more complex and larger, so they require higher computation, storage, and memory, which are typically constrained in edge networks but abundant in the cloud. Thus, the typical scenario is that the data owner (at the edge) is geographically separated from the model owner (in the cloud). 

The geographically separated nature of data and model owners poses challenges for ML inference as a service. We can naturally use a client/server-based approach where the data owner (client) sends its data to the model owner (cloud server) for ML inference. This approach violates data privacy and introduces communication overhead between the data and model owners, which is usually considered a bottleneck link in today's systems. Very promising privacy-preserving mechanisms have been investigated in the literature \cite{oblivious, gazelle, delphi}, which preserve the privacy of data in the client/server ML inference setup, but still suffer from high communication costs between data and model owners. Such high communication cost undermines the effectiveness of ML inference applications that require less latency and real-time response.

\begin{figure}
\centerline{\includegraphics[width=2.7in]{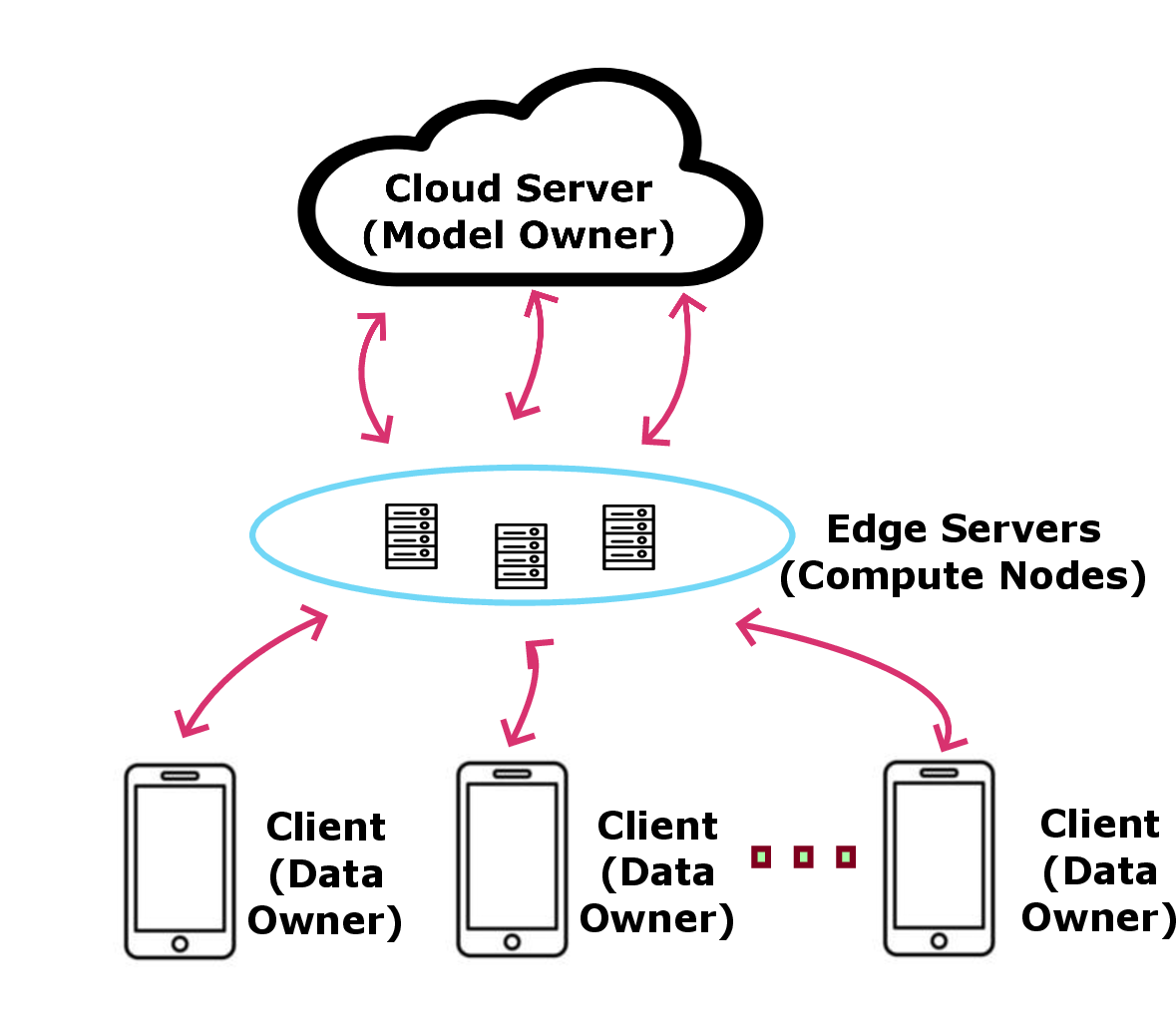}
}
\caption{Hierarchical ML inference.}
\label{fig:3Party}
 \end{figure}

A promising solution is a hierarchical setup, where clients own/generate data, model owners (cloud servers) have a pre-trained ML model, and compute nodes (edge servers) perform ML inference on clients' data using the cloud server's ML model, Fig. \ref{fig:3Party}. This approach advocates that edge servers perform ML inference by preserving the privacy of data from the client's perspective and ML model from the cloud server's perspective \cite{secureml, securenn, wagh2020falcon}. Hierarchical ML inference is promising to reduce the communication overhead between the client and cloud server by confining the communication cost between the client and edge servers, noting that data owners and compute nodes can communicate over high-speed edge links as they are geographically close. 

Despite the promise, the potential of hierarchical ML inference has not yet been fully explored in terms of utilizing available resources in edge servers. 
In this work, we consider (i) model-distributed inference (model parallelization) at the edge servers to speed up ML inference and (ii) reducing the amount of communication to/from the cloud server while preserving the privacy of both data and model. 


Model-distributed inference is emerging as a promising solution \cite{EdgePipe-hu2021pipeline, 10138654}, where an ML model is distributed across edge servers, Fig. \ref{fig:MDI}. The client transmits its data to an edge server, which processes a few layers of an ML model and transmits the feature vector of its last layer/block to the next edge server. Each edge server that receives a feature vector processes the layers that are assigned to it. The edge server that calculates the last layers of the ML model obtains and sends the output back to the client. We note that the edge servers perform parallel processing by data pipelining, so ML inference becomes faster. We consider that the edge servers in Fig. \ref{fig:3Party} could employ model parallelization for faster ML inference, but the crucial question of how to provide privacy to both data and the model should be addressed, which is the focus of this paper. 

\begin{figure}
\centerline{\includegraphics[width=2in]{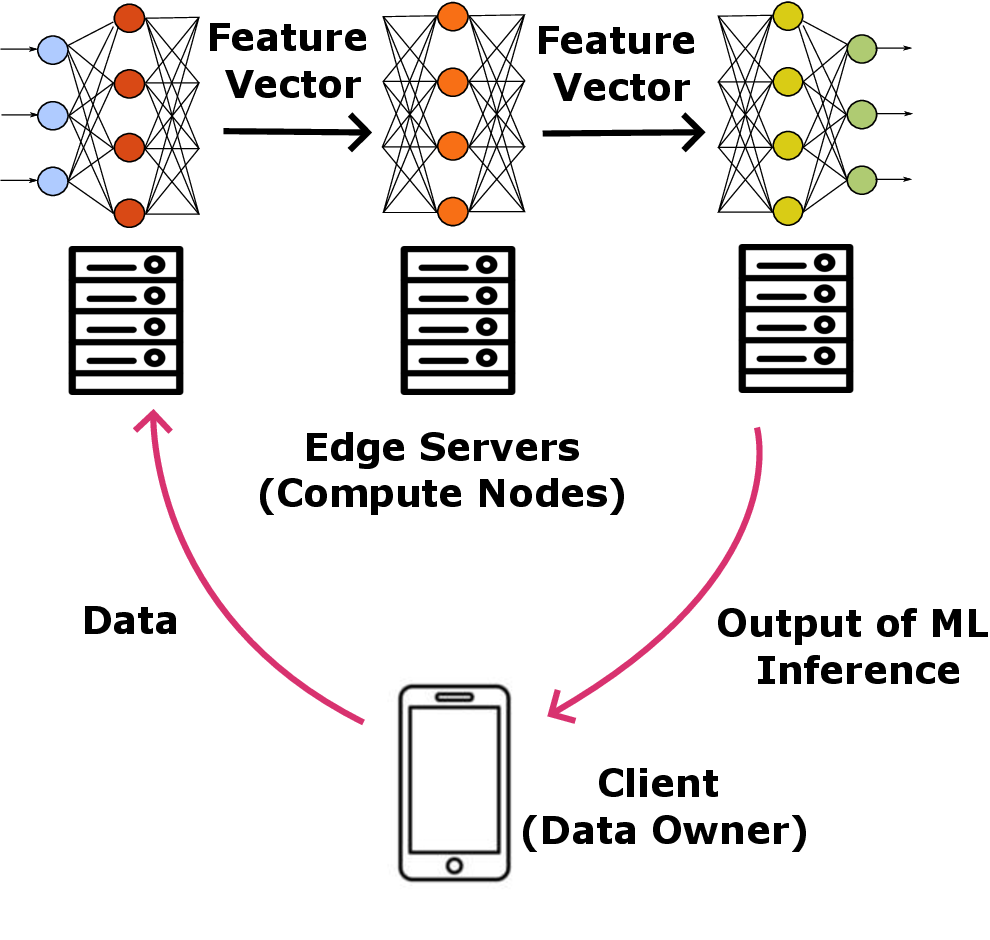}
}
\caption{Model-distributed inference.}
\label{fig:MDI}
 \end{figure}

In this paper, we design a privacy-preserving hierarchical model-distributed inference, \pmdi protocol. \pmdi uses hierarchical ML inference demonstrated in Fig. \ref{fig:3Party}. Similar work has been explored in the literature \cite{secureml, securenn, wagh2020falcon} but did not consider model-distributed inference, which is one of our novelties. In particular, compute nodes process parts of the ML model in parallel (thanks to pipelining) to reduce ML inference delay. 
Furthermore, we structure \pmdi in two phases: offline and online. The offline and online phases are designed in a way that the communication overhead of the online phase is small. Indeed, the communication cost of the online phase to/from the cloud server is zero in \pmdi. This significantly reduces the ML inference time of \pmdi. Our \pmdi design uses additive secret sharing and linearly homomorphic encryption to handle linear calculations in the ML inference, and garbled circuit and three-party oblivious transfer are used to handle non-linear functions (such as ReLU). The following are our contributions.

\begin{itemize}
    \item We design an ML inference protocol \pmdi that uses model-distributed inference to speed up ML inference and employs a hierarchical setup to reduce the communication overhead while providing privacy to both data and model. To the best of our knowledge, \pmdi is the first privacy-preserving model-distributed inference protocol in a hierarchical ML setup. 
    
    \item \pmdi consists of offline and online phases. We designed these phases in a way that most of the data exchange (hence communication) is done in the offline phase, which is done anytime before the online phase as it is independent of the client's data. Thus, the communication overhead of the online phase is reduced. Indeed, there is no communication to/from the cloud server in the online phase, and the amount of communication between the client and edge servers is minimized.
    
    \item Our \pmdi design uses additive secret sharing and linearly homomorphic encryption to handle linear calculations in the ML inference, and garbled circuit and three-party oblivious transfer are used to handle non-linear functions. \pmdi uses additive secret sharing with homomorphic encryption in the offline phase, which reduces the number of computations in the online phase significantly. Our novel three-party OT design, inspired by PROXY-OT protocol from \cite{naor1999privacy}, reduces the complexity of the existing OT protocols \cite{delphi, chameleon} and provides information-theoretic security.
    
    \item We implemented \pmdi as well as baselines using ACCESS computing platform \cite{NSFACCESS} and our in-lab computers. The experimental results demonstrate that \pmdi significantly improves the ML inference time as compared to the baselines. 
\end{itemize}

%% file: related_work.tex
\section{Related Work} \label{sec:related_work}

Privacy-preserving ML inference protocols can be roughly classified into two categories; (i) client/server protocols and (ii) hierarchical protocols as in Fig. \ref{fig:3Party}. \footnote{We note that we focus on privacy in this paper, so we do not consider active adversaries and do not include the related literature on active adversaries. In a similar spirit, we also exclude most of the protocols using trusted execution environments (TEEs), as our work does not rely on TEEs.}


\textbf{Client/server protocols.}
CryptoNets \cite{gilad2016cryptonets} pioneered the client/server protocols by using leveled homomorphic encryption (HE) and polynomial approximations for activation functions. CryptoDL \cite{hesamifard2017cryptodl} improved upon CryptoNets with better function approximations, while LoLa \cite{brutzkus2019low} focused on reducing latency in HE-based inference. In contrast, nGraph-HE \cite{boemer2019ngraph} avoids such approximations by passing feature maps to the data owner for clear text processing, raising concerns about ML model parameter privacy.

To address the challenges of ML model exposure and using approximation for the activation functions, subsequent studies combine HE with multi-party computation (MPC) techniques like oblivious transfer (OT), garbled circuits (GC), and secret sharing. For instance, MiniONN \cite{oblivious} and CrypTFlow2 \cite{rathee2020cryptflow2} integrate HE with MPC to enhance performance and accuracy. CrypTFlow2 uses OT or HE for linear layers and specialized MPC protocols for ReLU and Maxpool to handle non-linear activation functions. Similarly, Gazelle \cite{gazelle}, Delphi \cite{delphi}, AutoPrivacy \cite{lou2020autoprivacy}, and MP2ML \cite{boemer2020mp2ml} employ a combination of HE and garbled circuits to optimize both computational efficiency and privacy, using secret sharing for linear layers and GC for non-linear layers. Meanwhile, Cheetah \cite{huang2022cheetah} streamlines the process by eliminating costly HE operations, such as rotations in linear layer evaluations, and instead uses efficient OT-based protocols for handling non-linearities.

In parallel, a growing body of work aims to sidestep the high complexity of traditional cryptographic techniques like HE and GC by discretizing DNNs \cite{agrawal2019quotient, riazi2019xonn}. For example, XONN \cite{riazi2019xonn} simplifies the secure computation landscape by leveraging binary neural networks (BNNs), which operate on binary values (+1 or -1). This approach significantly reduces the cryptographic burden by focusing primarily on efficient XNOR operations. Similarly, COINN \cite{hussain2021coinn} utilizes advanced quantization techniques to lower both communication and computational overhead at the expense of accuracy loss.

Despite these advancements, client/server models often encounter high latency issues due to the bottleneck between the client and cloud server, especially when the participating parties are geographically distant—a common scenario in real-world applications. This latency can significantly affect the performance and feasibility of ML inference. Thus, we consider a hierarchical model in our setup, which differentiates us from client/server protocols. SECO \cite{chen2024seco} builds upon Delphi and splits an ML model into two parts, where one of them is processed at the client while the other part is processed at the cloud server, but it does not 
support model splitting and distribution over multiple (edge) servers. 



\textbf{Hierarchical protocols:} Hierarchical protocols for ML inference with the participation of three or more parties have been explored in \cite{secureml, chameleon, securenn, aby3} usually assuming that a client has data and multiple servers (edge and cloud) privately access or hold the ML model. 
ABY\textsuperscript{3} \cite{aby3} is a three-party framework that efficiently transitions between arithmetic, binary, and Yao's three-party computation (3PC), using a three-server model that tolerates a single compromised server. 
Cryptflow \cite{kumar2020cryptflow} introduces a three-party protocol built upon SecureNN \cite{securenn}, tolerating one corruption and optimizing convolution for reduced communication in non-linear layers. Falcon \cite{wagh2020falcon} combines techniques from SecureNN and ABY\textsuperscript{3} to improve protocol efficiency. SSNet \cite{duan2024ssnet} introduces a private inference protocol using Shamir's secret sharing scheme to handle collusion among multiple parties while providing privacy for both data and the ML model. 

\emph{Our work in perspective.} As compared to existing hierarchical protocols summarized above, our approach (i) uses additive secret sharing with HE to provide privacy for ML model parameters, which reduces the number of computations in the online phase significantly (for example, as compared to SSNet \cite{duan2024ssnet}); (ii) eliminates the need for continuous and interactive communication with the client during the online phase, thus reducing the communication overhead; (iii) uses model-distributed inference, which enables full utilization of available edge servers and their computing power, and (iv) works with any non-linear ML inference functions. 

%% file: system_model.tex
\section{System Model \& Preliminaries} \label{sec:system}

\textbf{Notations.} We define $\mathbb{Z}_q$ as a finite field of size $q$, and $\mathbf{x} \in \mathbb{Z}_q^n$ as a vector of size $n$ over the field $\mathbb{Z}_q$. Similarly, $\mathbf{x} \in \{0,1\}^n$ is defined as a binary vector of length $n$. We will denote vectors by bold lowercase letters, while matrices are denoted by bold uppercase letters.

\textbf{Setup.} We consider a three-party ML inference, which includes a cloud server (model owner), client (data owner), and edge servers (compute nodes), Fig. \ref{fig:3Party}. Clients and edge servers are directly connected, where high-speed device-to-device links can be used. Client and cloud servers can also communicate using infrastructure-based links such as Wi-Fi or cellular, which is typically considered a bottleneck in today's communication networks. This paper aims to reduce the communication cost between clients and the remote cloud server. 

Our system supports multiple clients, as shown in Fig. \ref{fig:3Party}, but we will focus on only one client in the rest of the paper for the sake of clarity and as the multiple client extension is straightforward. The edge servers are divided into clusters. Each cluster consists of $T+2$ edge servers, two of which are garbler and evaluator servers, which are needed to implement the garbled circuit operations, detailed later in this section. Each cluster is responsible for processing a set of layers according to model-distributed inference (model-distribution algorithm), i.e., each cluster ($T+2$ edge servers) computes a part of the model. We will present the details of cluster operation in Section \ref{sec:privateMDI}. 

\textbf{Threat Model.} We consider that all the participants are semi-honest, i.e., they follow the protocols, but they are curious. Edge servers in each cluster may collude; we consider that maximum $T$ edge servers (out of $T+2$) collude to obtain data. We assume that the evaluator and the garbler servers do not collude. Also, clients, edge servers, and the cloud server do not collude. Our aim is to design a privacy-preserving model distributed inference mechanism at the edge servers where (i) the client and edge servers do not learn anything about the model\footnote{We note that the client knows the number of non-linear layers and their dimensions in the ML mode, which is needed in our \pmdi design, but not the whole architecture of the ML model, nor the specific weights of the ML model.} and (ii) the cloud server and edge servers do not learn anything about the client's data.

\textbf{ML Model and Model-Distribution.} We consider that the cloud server stores a pre-trained ML model. We assume that the ML model can be partitioned, which applies to most of the ML models used in today's ML applications \cite{10138654}. The ML model could have any linear and/or non-linear operations. Our mechanism is designed to work with any such operations. 

Different clusters of edge servers may have different and time-varying computational capacities. Thus, a cluster with higher computing power should process more layers than the others. We follow a similar approach of AR-MDI proposed in \cite{10138654} to achieve such an adaptive model-distributed inference. 
AR-MDI \cite{10138654} is an ML allocation mechanism that determines the set of layers $\Lambda_n(k)$ that should be activated at edge server $n$ for processing data $A_k$. AR-MDI allocates $\lfloor{\rho_n(k)}\rceil$ layers to edge server $n$ for data $A_k$, where $\lfloor{.}\rceil$ rounds $\rho_n(k)$ to the closest integer that is a feasible layer allocation. AR-MDI determines $\rho_n(k)$ as $\rho_n(k) = W \frac{1/\gamma_n(k)}{\sum_{m=0}^{N-1} 1 / \gamma_m(k)}$, 
where $W$ is the total number of parameters in the ML model, $\gamma_n(k)$ is the per parameter computing delay. AR-MDI performs layer allocation decentralized as each worker can determine their share of layers by calculating $\lfloor{\rho_n(k)}\rceil$. The per parameter computing delay $\gamma_n(k)$ can be measured by each worker and shared with other workers. We note that we will use the AR-MDI's model allocation mechanism over our clusters of edge servers and by providing privacy for both data and model.

\textbf{Three-Party Oblivious Transfer.} Oblivious Transfer (OT) protocols \cite{naor1999oblivious, rivest1999unconditionally} typically consider two-party setup with a ``sender'' and ``receiver''. The idea is that the sender has two binary string inputs $\mathbf{k}_0$ and $\mathbf{k}_1$, and the receiver would like to learn $\mathbf{k}_i$, $i \in \{0,1\}$, but (i) the sender should not learn which input is selected, and (ii) the receiver learns only $\mathbf{k}_i$ and gains no information about $\mathbf{k}_{1-i}$. In our \pmdi design, we need to use three-party OT \cite{naor1999privacy} as detailed in Section \ref{sec:privateMDI}.

In particular, we design a novel three-party OT inspired by PROXY-OT introduced in \cite{naor1999privacy} providing information-theoretic privacy. In our three-party protocol shown in Fig. \ref{fig:proxy-ot}, the three parties are the client, evaluator server, and garbler server. The client has a one-bit input $i$, and the garbler server has the input labels $\mathbf{k}_0$ and $\mathbf{k}_1$. Additionally, the evaluator server and the client generate random sample $b$ using a pseudorandom generator with the same seed, and the evaluator and garbler server generate random samples $\mathbf{u}_0$ and $\mathbf{u}_1$ similarly. Here, $b, i \in \{0, 1\}, \mathbf{u}_j, \mathbf{k}_j \in \{0, 1\}^\kappa, j = 0, 1$.
First, the client sends its masked input $i \oplus b$ to the garbler server. Next, the garbler server sends the masked labels to the client. The client sends $\mathbf{k}_i \oplus \mathbf{u}_b$ to the evaluator server. The evaluator server can unmask the received label to obtain $\mathbf{k}_i$. The privacy proof of our three-party OT protocol is provided in Appendix A.


\begin{figure}[h]
    \centering
    \resizebox{\columnwidth}{!}{
        \begin{tikzpicture}[
            node distance=1.5cm and 1.5cm,
            >=Stealth,
            every node/.style={font=\footnotesize}
            ]

            \node (eval) at (0.5,0) {\textbf{Evaluator Server}};
            \node[below of=eval, node distance=0.5cm] (eval-inputs) {$\mathbf{u}_0,\mathbf{u}_1, b$};
            \node (garbler) at (4.5,0) {\textbf{Garbler Server}};
            \node[below of=garbler, node distance=0.5cm] (garbler-inputs) {$\mathbf{k}_0,\mathbf{k}_1,\mathbf{u}_0,\mathbf{u}_1$};
            \node (client) at (8,0) {\textbf{Client}};
            \node[below of=client, node distance=0.5cm] (client-inputs) {$i,b$};

            \draw[->] (client) (8,-1.5) -- (4.5,-1.5) (garbler) node[midway, above] {$i \oplus b$};
            \draw[->] (garbler) (4.5,-2.25) -- (8,-2.25) (client) node[midway, above] {$\mathbf{k}_0 \oplus \mathbf{u}_{i \oplus b}$} node[midway, below] {$\mathbf{k}_1 \oplus \mathbf{u}_{\overline{i \oplus b}}$};
            \draw[->] (client) (8,-3) -- (0.5,-3) (eval) node[midway, above] {$\mathbf{k}_i \oplus \mathbf{u}_b$};

            \draw[thick] (-1,0.5) rectangle (9,-3.5);

        \end{tikzpicture}
    }
    \caption{Our novel three-party OT protocol.}
    \label{fig:proxy-ot}
\end{figure}
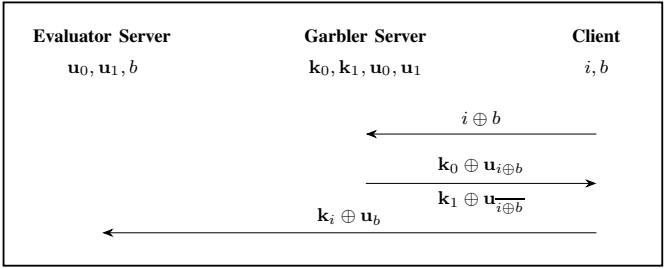

\textbf{Garbled Circuit.} Garbled Circuit (GC) is a technique for encoding a boolean circuit $C$ and its inputs $\mathbf{x}$ and $\mathbf{y}$ in a way that allows an evaluator to compute the output $C(\mathbf{x},\mathbf{y})$ without revealing any information about $C$ or $\mathbf{x}$ and $\mathbf{y}$ other than the output itself. This process involves a garbling scheme consisting of algorithms for encoding and evaluating the circuit, ensuring completeness (the output matches the actual computation) and privacy (the evaluator learns nothing beyond the output and the size of the circuit).

A garbling scheme \cite{Yao1986HowTG, bellare2012foundations} is a tuple of algorithms \texttt{GS $=$ (Garble,Eval)} with the following syntax:
\begin{itemize}
    \item $(\Tilde{C},\{\mathbf{k}_{0}^j,\mathbf{k}_{1}^j\}_{j \in [2n]}) \leftarrow$ \texttt{GS.Garble$(1^\lambda,C)$}. Given a security parameter $\lambda$ and a boolean circuit $C$ as input, \texttt{Garble} produces a garbled circuit $\Tilde{C}$ and a set of labels $\{\mathbf{k}_{0}^j,\mathbf{k}_{1}^j\}_{j \in [2n]}$. Here, $\mathbf{k}_{i}^j$ represents assigning the value $i \in \{0,1\}$ to the $j$-th input label.
    
    \item Given a garbled circuit $\Tilde{C}$ and labels $\{\mathbf{k}_{\mathbf{x}_j}^j\}_{j \in [n]}$ corresponding to an input $\mathbf{x} \in \{0,1\}^n$, and labels $\{\mathbf{k}_{\mathbf{y}_j}^{n+j}\}_{j \in [n]}$ corresponding to an input $\mathbf{y} \in \{0,1\}^n$, \texttt{Eval} outputs a string $C(\mathbf{x},\mathbf{y})$.
\end{itemize}

\emph{Correctness:} For correctness, we require that for every circuit $C$ and inputs $\mathbf{x},\mathbf{y} \in \{0,1\}^n$ the output of \texttt{Eval} must equal $C(\mathbf{x},\mathbf{y})$. Formally, $\Prob\left[C(\mathbf{x},\mathbf{y}) = \texttt{Eval}\left(\Tilde{C},\{\mathbf{k}_{\mathbf{x}_j}^j,\mathbf{k}_{\mathbf{y}_j}^{n+j}\}_{j \in [n]}\right)\right] = 1$.

\emph{Security:} For security, we require that a simulator $\mathcal{S}_{\textnormal{GS}}$ exists such that for any circuit $C$ and inputs $\mathbf{x},\mathbf{y} \in \{0,1\}^n$, we have $\left( \Tilde{C},\{\mathbf{k}_{\mathbf{x}_j}^j,\mathbf{k}_{\mathbf{y}_j}^{n+j}\}_{j \in [n]} \right) \approx_c \mathcal{S}_{\text{GS}}\left( 1^\lambda,1^{|C|},C(\mathbf{x},\mathbf{y}) \right)$.  

Our \pmdi design involves three parties: the garbler with input $\mathbf{x}$, the evaluator without any input, and the client with input $\mathbf{y}$. For the evaluator to compute $C(\mathbf{x},\mathbf{y})$ without any party gaining information about other parties' inputs, conventional two-party OT is substituted with our three-party OT protocol.

\textbf{Linearly Homomorphic Encryption.} We use Linearly homomorphic encryption (LHE) in our \pmdi design, which enables certain computations on encrypted data without the need for decryption \cite{paillier1999public, gentry2009fully}. In LHE, operations on ciphertexts correspond to linear operations on the plaintexts they encrypt. A linearly homomorphic encryption scheme is characterized by four algorithms, collectively denoted as \texttt{HE = (KeyGen,Enc,Dec,Eval)}, which can be described as
\begin{itemize}
    \item \texttt{HE.KeyGen} is an algorithm that generates a public key \texttt{pk} and a secret key \texttt{sk} pair.
    
    \item \texttt{HE.Enc(pk,$m$)} encrypts the message $m$ using the public key \texttt{pk} and outputs a ciphertext $ct$. The message space is a finite field $\mathbb{Z}_q$.
    
    \item \texttt{HE.Dec(sk,$ct$)} decrypts the ciphertext $ct$ using the secret key \texttt{sk} and outputs the message $m$.
    
    \item \texttt{HE.Eval(pk,$ct_1$,$ct_2$,$f$)} outputs the encrypted $f(m_1,m_2)$ using the public key \texttt{pk} on the two ciphertexts $ct_1$ and $ct_2$ encrypting messages $m_1$ and $m_2$, where $f$ is a linear function.
\end{itemize}

Let $ct_1 = \texttt{HE.Enc}(pk, m_1)$, $ct_2 = \texttt{HE.Enc}(pk, m_2)$, $ct' = \texttt{HE.Eval}(pk, ct_1, ct_2, f)$. We require \texttt{HE} to satisfy the following properties:
\begin{itemize}
    \item \emph{Correctness:} \texttt{HE.Dec} outputs $m$ using \texttt{sk} and a ciphertext $ct =$ \texttt{HE.Enc(pk,$m$)}.

    \item \emph{Homomorphism:} \texttt{HE.Dec} outputs $f(m_1,m_2)$ using \texttt{sk} and a ciphertext $ct' = \texttt{HE.Eval}(pk, ct_1, ct_2, f)$.
    
    \item \emph{Semantic security:} Given a ciphertext $ct$ and two messages of the same length, no attacker should be able to tell which message was encrypted in $ct$.
    
    \item \emph{Function privacy:} Given a ciphertext $ct$, no attacker can tell what homomorphic operations led to $ct$. Formally, $\left( ct_1,ct_2,ct' \right)
        \approx_c
        \mathcal{S}_{\textnormal{FP}}(1^\lambda,m_1,m_2,f(m_1,m_2))$.
\end{itemize}

\textbf{Additive Secret Sharing.} An $m$-of-$m$ additive secret sharing scheme over a finite field $\mathbb{Z}_q$, splits a secret $\mathbf{x} \in \mathbb{Z}_q$ into a $m$-element vector $([\mathbf{x}]_1, \ldots , [\mathbf{x}]_m) \in \mathbb{Z}_q^m$. The scheme consists of a pair of algorithms \texttt{ADD = (Shr,Re)}, where \texttt{Shr} corresponds to share, and \texttt{Re} refers to reconstruct; 
\begin{itemize}
    \item $(\mathbf{a}_1,\ldots,\mathbf{a}_m) \leftarrow \texttt{ADD.Shr}(\mathbf{x},m)$. On inputs a secret $\mathbf{x} \in \mathbb{Z}_q$ and number of shares $m$, \texttt{Shr} samples $(m-1)$ values $\mathbf{a}_1,\ldots,\mathbf{a}_{m-1}$ and computes $\mathbf{a}_m = \mathbf{x} - \sum_{j = 1}^{m-1} \mathbf{a}_j$. Then it outputs the shares $(\mathbf{a}_1,\ldots,\mathbf{a}_m)$.

    \item The \texttt{ADD.Re} algorithm on input a sharing $(\mathbf{a}_1,\ldots,\mathbf{a}_m)$, computes and outputs $\mathbf{x} = \sum_{j = 1}^m \mathbf{a}_j$.
\end{itemize}

\emph{Correctness:} Let \( \mathbf{x} \in \mathbb{Z}_q \) be a secret. Then: $\texttt{ADD.Re}(\texttt{ADD.Shr}(\mathbf{x})) = \mathbf{x} $. 

\emph{Security:} Let $(\mathbf{a}_1^1, \ldots, \mathbf{a}_m^1) \leftarrow \texttt{ADD.Shr}(\mathbf{x}_1, m)$ and $(\mathbf{a}_1^2, \ldots, \mathbf{a}_m^2) \leftarrow \texttt{ADD.Shr}(\mathbf{x}_2, m)$ be the secret shares of two secrets, $\mathbf{x}_1$ and $\mathbf{x}_2$, respectively. Here, the $j$-th party possesses the shares $\mathbf{a}_j^1$ and $\mathbf{a}_j^2$. For any $j$, the shares $\mathbf{a}_j^1$ and $\mathbf{a}_j^2$ are identically distributed. Consequently, each share is indistinguishable from a random value to the $j$-th party, ensuring no information about $\mathbf{x}_1$ or $\mathbf{x}_2$ is leaked.

%% file: overall_offline_protocol.tex
\section{Privacy-Preserving Model-Distributed Inference (\pmdi) Design} \label{sec:privateMDI}

The objective of our \pmdi protocol is to compute the inference result of the client's input data $\mathbf{x} \in \mathbb{Z}_q^n$ using an ML model $\mathbf{M}(.)$ and give output $\mathbf{M}(\mathbf{x})$ in a privacy-preserving manner. The ML model consists of $L$ layers and is represented by $\mathbf{M} = (\mathbf{M}_1,\ldots, \mathbf{M}_L)$ by abuse of notation.

\begin{algorithm}[ht]
 \begin{algorithmic}[1]
    \Input The cloud server has the ML model parameters $\mathbf{M} = (\mathbf{M}_1, \ldots, \mathbf{M}_L)$.
        \State  The client generates \texttt{(pk,sk)} using \texttt{HE.KeyGen}.
        
        \State  The client generates random sample $\mathbf{r}_1$ over $\mathbb{Z}_q^n$, and sends \texttt{HE.Enc(pk,$\mathbf{r}_{1}$)} to the cloud server.
        
        \State  The cloud server initializes \( \mathbf{c}_1^{enc} \gets \text{\texttt{HE.Enc(pk,$\mathbf{r}_{1}$)}} \).
        
    \For {$j \in [1, \ldots , P]$}
        \For {$l \in [l_{start}^j, \ldots ,l_{end}^j]$}
            \State Cloud server generates random sample $\mathbf{s}_l \in \mathbb{Z}_q^n$,  secret shares $([\mathbf{M}_l]_1$,..., $[\mathbf{M}_l]_{T + 1})$ $\gets$ $\texttt{ADD.Shr}$ $(\mathbf{M}_l,T+1)$ and  $([\mathbf{s}_l]_1$, $\ldots$ , $[\mathbf{s}_l]_{T + 1})$ $\gets$ $\texttt{ADD.Shr}(\mathbf{s}_l,T+1)$.
            \State The cloud server sends $[\mathbf{M}_l]_v$ and $[\mathbf{s}_l]_v$ to the edge server $v$, $v \in [v^j_1,\ldots,v^j_{T+1}]$.
            \State The cloud server encrypts: $\mathbf{M}_l^{enc} \gets $ \texttt{HE.Enc(pk,$\mathbf{M}_l$)}, $\mathbf{s}_l^{enc} \gets $ \texttt{HE.Enc(pk,$\mathbf{s}_l$)}.
            \State The cloud server computes $\mathbf{M}_l^{enc} \cdot \mathbf{c}_l^{enc} + \mathbf{s}_l^{enc}$.
            \If {$l$ only has linear operations}
                \State The cloud server determines $\mathbf{c}_{l+1}^{enc} \gets \mathbf{M}_l^{enc} \cdot \mathbf{c}_l^{enc} + \mathbf{s}_l^{enc}$.
            \Else 
                \State The cloud server sends $\mathbf{M}_l^{enc} \cdot \mathbf{c}_l^{enc} + \mathbf{s}_l^{enc}$ to the client.
                \State The client sends \texttt{HE.Enc(pk,$\mathbf{r}_{l+1}$)} to the cloud server.
                \State The cloud server determines $\mathbf{c}_{l+1}^{enc} \gets $ \texttt{HE.Enc(pk,$\mathbf{r}_{l+1}$)}.
                \State The client, garbler, and evaluator server run Algorithm \ref{alg:non-linear_offline_algorithm}.
            \EndIf
        \EndFor
    \EndFor
    \If {$L$ only has linear operations}
        \State The cloud server sends $\mathbf{c}_{L+1}^{enc}$ to the client.
    \EndIf
 \end{algorithmic}
 \caption{Offline \pmdi operation} \label{alg:overall_offline_algorithm}
\end{algorithm}

Our \pmdi protocol is divided into two phases: offline and online, similar to some previous work \cite{secureml, chameleon, aby3, oblivious, gazelle, delphi}. The goal behind having two phases is to move the computationally intensive aspects to the offline phase and make the online phase, which calculates the ML model output, faster. Also, the offline phase is designed to minimize the communication overhead between the client and the cloud server, which is the bottleneck link in the online phase. The offline phase (i) exchanges keys between the client and cloud server, (ii) shares the secret ML model with the edge servers, and (iii) exchanges the garbled circuits and two out of three inputs of the garbled circuits labels. We note that the offline phase does not use and is independent of the client's data $\mathbf{x}$, but makes the system ready for the online phase, i.e., computing the ML inference $\mathbf{M}(\mathbf{x})$ in a privacy-preserving manner. 

We note that an ML model consists of linear and non-linear layers. In our \pmdi design, we use additive secret sharing, linear homomorphic encryption for linear operations, garbled circuit, and three-party oblivious transfer for non-linear operations. As we mentioned earlier, our \pmdi design is generic enough to work with any non-linear functions. Next, we will describe the details of the offline and online parts of the \pmdi.

\subsection{Offline \pmdi} \label{subsec:overall_alg_offline}
In this section, we describe the details of the offline \pmdi protocol. The overall operation is summarized in Algorithm \ref{alg:overall_offline_algorithm}. 

The cloud server has the ML model parameters $\mathbf{M} = (\mathbf{M}_1,\ldots, \mathbf{M}_L)$. In the initialization phase of the algorithm (lines 1-2), the client generates a pair of \texttt{(pk,sk)} using \texttt{HE.KeyGen} and a random sample $\mathbf{r}_1$ over $\mathbb{Z}_q^n$, and sends \texttt{HE.Enc(pk,$\mathbf{r}_1$)} to the cloud server. We note that $\mathbf{r}_l,\mathbf{s}_l \in \mathbb{Z}_q^n$ are random samples of length $n$ generated by the client and cloud server for ML model layer $l$, respectively. The cloud server initializes $\mathbf{c}_1^{enc}$ with the received ciphertext $\mathbf{r}_1$ from the client according to \( \mathbf{c}_1^{enc} \gets \text{\texttt{HE.Enc(pk,$\mathbf{r}_{1}$)}} \).

Next, the offline protocol shares encrypted model parameters and keys (garbles circuit labels) for each cluster of edges, where there are $P$ clusters. We note that there are $T+2$ edge servers in each cluster. Each cluster $j$ processes the set of layers $[l_{start}^j, \ldots ,l_{end}^j]$ that are assigned to cluster $j$ employing the adaptive model-distribution algorithm AR-MDI described in Section \ref{sec:system}, where $l_{start}^j$ and $l_{end}^j$ are the first and last assigned layers to cluster $j$. The details of the offline protocol for linear and non-linear ML operations are provided in the following.  

\begin{algorithm}[ht]
 \begin{algorithmic}[1]
    \Input Circuit $C$.
    \State The garbler server runs the \texttt{GS.Garble} algorithm on circuit $C$ and outputs $\Tilde{C}$ and the garbled labels.
    \State The garbler server sends $\Tilde{C}$ to the evaluator server.
    \State The garbler server, evaluator server, and client run the three-party OT algorithm.
    \State The evaluator server receives the garbled labels of $\mathbf{M}_l \cdot \mathbf{c}_l + \mathbf{s}_l$ and $\mathbf{r}_{l+1}$.
 \end{algorithmic}
 \caption{Garbled circuit and oblivious transfer operation in a cluster in the \emph{offline} phase.} \label{alg:non-linear_offline_algorithm}
\end{algorithm}

The cloud server generates random sample $\mathbf{s}_l$ over $\mathbb{Z}_q^n$, calculates $T + 1$ additive secret shares of the model parameters $\mathbf{M}_l$ and the random samples $\mathbf{s}_l$ (line 6). We note that $\mathbf{s}_l$ provides privacy for the ML model parameters against client and edge servers. Then, it sends each share ($[\mathbf{M}_l]_v$ and $[\mathbf{s}_l]_v$) to the edge server $v$ of cluster $j$ (line 7). We note that $v$ is any edge server in cluster $j$ such that $v \in [v^j_1,\ldots,v^j_{T+1}]$ excluding the evaluator server. 
%
The cloud server encrypts $\mathbf{M}_l$ and $\mathbf{s}_l$ and determines  the enycrpted versions $\mathbf{M}_l^{enc} \gets $ \texttt{HE.Enc(pk,$\mathbf{M}_l$)}, $\mathbf{s}_l^{enc} \gets $ \texttt{HE.Enc(pk,$\mathbf{s}_l$)} (line 8).
%
%
%
The cloud server computes $\mathbf{M}_l^{enc} \cdot \mathbf{c}_l^{enc} + \mathbf{s}_l^{enc}$ by executing the \texttt{HE.Eval} algorithm on $\mathbf{M}_l^{enc}$, $\mathbf{c}_l^{enc}$, and $\mathbf{s}_l^{enc}$ (line 9).

Next, the offline \pmdi determines $\mathbf{c}_{l+1}^{enc}$, which is a secret share needed in the online part. Depending on whether a layer of the ML model has only linear or non-linear components, the calculation of $\mathbf{c}_{l+1}^{enc}$ differs. If a layer only has \textbf{linear} components,  the cloud server determines $\mathbf{c}_{l+1}^{enc} \gets \mathbf{M}_l^{enc} \cdot \mathbf{c}_l^{enc} + \mathbf{s}_l^{enc}$ (line 11).
    
The following steps (lines 13-16) are performed if layer $l$ has \textbf{non-linear} components. First, the cloud server sends $\mathbf{M}_l^{enc} \cdot \mathbf{c}_l^{enc} + \mathbf{s}_l^{enc}$, computed in the linear part, to the client. Then the client sends a new random sample \texttt{HE.Enc(pk,$\mathbf{r}_{l+1}$)} to the cloud server, and the cloud server determines $\mathbf{c}_{l+1}^{enc}$ with \texttt{HE.Enc(pk,$\mathbf{r}_{l+1}$)}. Finally, the garbled circuit is used to compute the non-linear activation function privately. 

In particular, cluster $j$'s garbler and evaluator servers and the client run Algorithm \ref{alg:non-linear_offline_algorithm}. Using this algorithm, the garbler server in the cluster runs the \texttt{GS.Garble} algorithm on the input circuit $C$ and sends the output to the evaluator server (lines 1-2 of Algorithm \ref{alg:non-linear_offline_algorithm}). Next, the client, garbler, and evaluator server run the three-party OT algorithm described in Section \ref{sec:system} (line 3 of Algorithm \ref{alg:non-linear_offline_algorithm}). At the end, the evaluator receives the labels corresponding to $\mathbf{M}_l \cdot \mathbf{c}_l + \mathbf{s}_l$ and $\mathbf{r}_{l+1}$, two out of three inputs of the garbled circuit, where $\mathbf{M}_l \cdot \mathbf{c}_l + \mathbf{s}_l$ is the decrypted value of $\mathbf{M}_l^{enc} \cdot \mathbf{c}_l^{enc} + \mathbf{s}_l^{enc}$. Note that the circuit $C$ is a boolean circuit.  Assuming that the non-linearity is due to the ReLU function as an example (noting that our algorithm works with any non-linear functions),  the circuit $C$ computes the results in the following order. 
\begin{enumerate}
    \item  $\mathbf{M}_l \cdot \mathbf{x}_l = (\mathbf{M}_l \cdot \mathbf{c}_l + \mathbf{s}_l) + (\mathbf{M}_l(\mathbf{x}_l - \mathbf{c}_l) - \mathbf{s}_l)$, where the second term (i.e., $(\mathbf{M}_l(\mathbf{x}_l - \mathbf{c}_l) - \mathbf{s}_l)$) is calculated in the online phase.
    \item  $\mathbf{x}_{l+1} = \text{ReLU}(\mathbf{M}_l \cdot \mathbf{x}_l)$.
    \item  $\mathbf{x}_{l+1} - \mathbf{r}_{l+1}$.
\end{enumerate}

In the last step of Algorithm \ref{alg:overall_offline_algorithm}, the cloud server sends the ciphertext $\mathbf{c}_{L+1}^{enc}$ to the client (line 21) if the last layer is linear, as it's needed for the client in the online phase to decrypt the inference result. If the last layer is non-linear, the client already possesses the decrypted $\mathbf{c}_{L+1}^{enc}$ according to line 14.

%% file: overall_online_protocol.tex
\subsection{Online \pmdi} \label{subsec:overall_alg_online}
The online phase of \pmdi is responsible for computing the ML inference function $\mathbf{M}(\mathbf{x})$ given the client's data $\mathbf{x}$. We note that there is no communication between edge servers and the cloud server or the client, except for the beginning of the algorithm, when the client sends the input data to the first cluster to start the process, and at the end when the last cluster sends the result of the inference to the client. The online phase is summarized in Algorithm \ref{alg:overall_online_algorithm}.

First, the client sends a secret share of its input $\mathbf{x}$ masked by its random sample $\mathbf{r}_1$ to the garbler server of the first cluster, denoted as $\mathbf{x}_1 - \mathbf{c}_1$, where $\mathbf{x}_1 \gets \mathbf{x}$ and $\mathbf{c}_1 \gets \mathbf{r}_1$ (line 1). 
Then, each layer $l$ assigned to cluster $j$ is processed with input $\mathbf{x}_l - \mathbf{c}_l$, where $\mathbf{c}_l$ is the decrypted value of $\mathbf{c}_l^{enc}$, determined in the offline phase. 
We note that the input of every layer is denoted by $\mathbf{x}_l - \mathbf{c}_l$, where $\mathbf{x}_l$ is the output of the previous $l-1$ layers. 

The cluster $j$ performs linear computations of layer $l$ (if there is any) 
on input $\mathbf{x}_l - \mathbf{c}_l$ to produce the output $\mathbf{M}_l(\mathbf{x}_l - \mathbf{c}_l) - \mathbf{s}_l$. 

First, the garbler server sends the input of the current layer $\mathbf{x}_l - \mathbf{c}_l$ to the other $T$ computing edge servers. Next, each computing edge server $v$, including the garbler server, uses $[\mathbf{M}_l]_v$ and $[\mathbf{s}_l]_v$ and computes $[\mathbf{M}_l]_v(\mathbf{x}_l - \mathbf{c}_l) - [\mathbf{s}_l]_v$. Note that each computing edge server makes its output private from the garbler server by adding $[\mathbf{s}_l]_v$ to its computation. Then, the garbler server runs the \texttt{ADD.Re} algorithm on the $T$ secret shares received from the computing edge servers and its own share $[\mathbf{M}_l]_{T+1}(\mathbf{x}_l - \mathbf{c}_l) - [\mathbf{s}_l]_{T+1}$, and obtains $\mathbf{M}_l(\mathbf{x}_l - \mathbf{c}_l) - \mathbf{s}_l$, which is the input of the next layer, denoted as $\mathbf{x}_{l+1} - \mathbf{c}_{l+1}$. 

Note that for a purely linear layer, if we expand $\mathbf{M}_l(\mathbf{x}_l - \mathbf{c}_l) - \mathbf{s}_l$, we get: $\mathbf{M}_l \cdot \mathbf{x}_l - \mathbf{M}_l \cdot \mathbf{c}_l - \mathbf{s}_l = \mathbf{M}_l \cdot \mathbf{x}_l - (\mathbf{M}_l \cdot \mathbf{c}_l + \mathbf{s}_l)$. \\
According to the definitions of $\mathbf{x}_{l+1}$ and $\mathbf{c}_{l+1}$, we have: $\mathbf{M}_l(\mathbf{x}_l - \mathbf{c}_l) - \mathbf{s}_l = \mathbf{M}_l \cdot \mathbf{x}_l - (\mathbf{M}_l \cdot \mathbf{c}_l + \mathbf{s}_l) = \mathbf{x}_{l+1} - \mathbf{c}_{l+1}$, which represents the input of the next layer, layer $l+1$.

\begin{algorithm}[ht]
 \begin{algorithmic}[1]
    \Input The client's input data $\mathbf{x}$ and the cloud server's model parameters $\mathbf{M}$.  
    \State The client sends a secret share of its input $\mathbf{x}$ masked by its random share $\mathbf{r}_1$ to the garbler server of the first cluster, denoted as $\mathbf{x}_1 - \mathbf{c}_1$, where $\mathbf{x}_1 \gets \mathbf{x}$ and $\mathbf{c}_1 \gets \mathbf{r}_1$.
    \For {$j \in [1, \ldots , P]$}
        \For {$l \in [l_{start}^j, \ldots ,l_{end}^j]$}
            \State The cluster $j$ performs linear operations on $\mathbf{x}_l - \mathbf{c}_l$ and determines $\mathbf{M}_l(\mathbf{x}_l - \mathbf{c}_l) - \mathbf{s}_l$
            \If {$l$ has non-linear components}
                \State The garbler and evaluator server run Algorithm \ref{alg:Cluster_online_non-linear_algorithm}, and outputs $\mathbf{x}_{l+1} - \mathbf{c}_{l+1}$, where $\mathbf{c}_{l+1} = \mathbf{r}_{l+1}$.
            \EndIf
        \EndFor
        \State The garbler server sends $\mathbf{x}_{l+1} - \mathbf{c}_{l+1}$ to the garbler server of the next cluster.
    \EndFor
    \State The garbler server of the last cluster sends $\mathbf{x}_{L+1} - \mathbf{c}_{L+1}$ to the client.
    \State The client unmasks the above secret share and obtains the inference result $\mathbf{M}(\mathbf{x}) = \mathbf{x}_{L+1}$. 
 \end{algorithmic}
  \textbf{Output:} Result of the inference $\mathbf{M}(\mathbf{x})$ on the client side. 
 \caption{Online \pmdi operation.} \label{alg:overall_online_algorithm}
\end{algorithm}

If $l$ has non-linear components, the garbler and evaluator server run Algorithm \ref{alg:Cluster_online_non-linear_algorithm}. This algorithm requires three sets of inputs: $\mathbf{M}_l(\mathbf{x}_l - \mathbf{c}_l) - \mathbf{s}_l$, calculated by the garbler server, $\mathbf{r}_{l+1}$, a random sample generated by the client to ensure the layer output's privacy, and  $\mathbf{M}_l \cdot \mathbf{c}_l + \mathbf{s}_l$, calculated by the cloud server and sent to the client in the offline phase.

The evaluator server obtains two sets of labels in the offline phase via the three-party OT protocol in Fig. \ref{fig:proxy-ot}. The garbler server sends the last set of labels corresponding to its output in the linear part, denoted as $\mathbf{M}_l(\mathbf{x}_l - \mathbf{c}_l) - \mathbf{s}_l$ (line 1 of Algorithm \ref{alg:Cluster_online_non-linear_algorithm}). We don't need OT to exchange the labels of this input. This justifies the garbler server's choice as the cluster's head server, responsible for communicating with the other edge servers of the cluster. Otherwise, it would increase the amount of communication. Then, with all the required labels and circuits, the evaluator server runs the \texttt{GS.Eval} algorithm and calculates $C(\mathbf{M}_l \cdot \mathbf{x}_l) - \mathbf{r}_{l+1}$, denoted by $\mathbf{x}_{l+1} - \mathbf{r}_{l+1}$ (line 2 of Algorithm \ref{alg:Cluster_online_non-linear_algorithm}). Finally, the evaluator server sends the result ($\mathbf{x}_{l+1} - \mathbf{r}_{l+1}$) back to the garbler server (line 3 of Algorithm \ref{alg:Cluster_online_non-linear_algorithm}).

\begin{algorithm}[ht]
 \begin{algorithmic}[1]
    \Input The garbled labels of $\mathbf{M}_l \cdot \mathbf{c}_l + \mathbf{s}_l$ and $\mathbf{r}_{l+1}$ stored by the evaluator server.
    \State The garbler server sends the corresponding labels of the input $\mathbf{M}_l(\mathbf{x}_l - \mathbf{c}_l) - \mathbf{s}_l$ to the evaluator server.
    \State The evaluator server runs \texttt{GS.Eval} and calculates $C(\mathbf{M}_l \cdot \mathbf{x}_l) - \mathbf{r}_{l+1}$, denoted by $\mathbf{x}_{l+1} - \mathbf{r}_{l+1}$.
    \State The evaluator server sends the result back to the garbler server.
    \Output $\mathbf{x}_{l+1} - \mathbf{r}_{l+1}$. 
 \end{algorithmic}
 \caption{Handling non-linearity in the \emph{online} phase.} \label{alg:Cluster_online_non-linear_algorithm}
\end{algorithm}

Cluster $j$ performs all the linear and non-linear parts of all the layers assigned to it in Algorithm \ref{alg:overall_online_algorithm}. Ultimately, the garbler server forwards the cluster's result either to the next cluster if $l_{end}^j$ is the final layer allocated to cluster $j$ (line 9) or the client (line 11) if $l_{end}^j$ is the ML model's final layer $L$, i.e., $l_{end}^j = L$. Finally, the client unmasks the secret share $\mathbf{x}_{L+1} - \mathbf{c}_{L+1}$ and obtains the inference result $\mathbf{M}(\mathbf{x}) = \mathbf{x}_{L+1}$.

%% file: analysis.tex
\section{\label{sec:analysis} Privacy and Communication Analysis} 
In this section, we analyze \pmdi in terms of its privacy guarantee and communication overhead. 

\subsection{Privacy} 

In this section, we establish the privacy properties of our protocol and demonstrate its privacy guarantees. We prove that our protocol is secure based on Definition \ref{def:privacy_protocol}. We use simulation-based, hybrid arguments to prove security \cite{goldreich2009foundations, fischlin2021overview}. 


\begin{definition}\label{def:privacy_protocol}
A cryptographic inference protocol $\mathbf{\Pi}$ involves a cloud server with model parameters $\mathbf{M} = (\mathbf{M}_1, \mathbf{M}_2, \ldots, \mathbf{M}_l)$, a client with an input vector $\mathbf{x}$, and $P$ clusters, each containing $T$ computing edge servers, a garbler server, and an evaluator server. The protocol is considered secure if it satisfies the following conditions:

\noindent \underline{Correctness:} The protocol ensures that for any given set of model parameters $\mathbf{M}$ held by the cloud server and any input vector $\mathbf{x}$ provided by the client, the client receives the correct prediction $\mathbf{M}(\mathbf{x})$ after executing the protocol.

\noindent \underline{Security:} Consider the following scenarios: 

\noindent \textbf{Compromised cloud server:} A semi-honest cloud server, even if partially compromised, should not learn the client's private input $\mathbf{x}$. This is formally captured by the existence of an efficient simulator $\mathcal{S}_{\textnormal{CS}}$ such that $\text{View}^{\mathbf{\Pi}}_{\textnormal{CS}} \approx_c \mathcal{S}_{\textnormal{CS}}(\mathbf{M})$, where $\text{View}^{\mathbf{\Pi}}_{\textnormal{CS}}$ denotes the cloud server's view during the protocol execution.   

\noindent  \textbf{Compromised client:} A semi-honest client, even if compromised, should not learn the cloud server's model parameters $\mathbf{M}$. This is formally captured by the existence of an efficient simulator $\mathcal{S}_{\textnormal{C}}$ such that $\text{View}^{\mathbf{\Pi}}_{\textnormal{C}} \approx_c \mathcal{S}_{\textnormal{C}}(\mathbf{x}, \mathbf{y})$, where $\text{View}^{\mathbf{\Pi}}_{\textnormal{C}}$ denotes the client's view during the protocol execution, including the input, randomness, and protocol transcript, and $\mathbf{y}$ is the output of the inference.
    
\noindent  \textbf{Compromised edge server:} Semi-honest edge servers within a cluster should not gain information about the client's input $\mathbf{x}$ or the cloud server's model parameters $\mathbf{M}$. This is formally captured by the existence of an efficient simulator $\mathcal{S}_{\textnormal{ES}}$ for each edge server $i$, such that $\text{View}^{\mathbf{\Pi}}_{\textnormal{ES},i} \approx_c \mathcal{S}_{\textnormal{ES},i}$, where $\text{View}^{\mathbf{\Pi}}_{\textnormal{ES},i}$ denotes the view of the edge server during the protocol execution. The edge server can be a computing edge server, a garbler server, or an evaluator server, $i \gets C, G, E$.

\end{definition}

\begin{theorem}\label{th:privacy}
\pmdi is secure according to Definition \ref{def:privacy_protocol} assuming the use of secure garbled circuits, linearly homomorphic encryption, and three-party OT.
\end{theorem}

\begin{proof}
We use simulation-based, hybrid arguments to prove Theorem \ref{th:privacy}
as detailed in Appendix B. 
\end{proof}

\begin{table*}[t]
    \caption{Communication overhead analysis. CS: Cloud Server (Model Owner), C: Client (Data Owner), and ES: Edge Server}
    \label{tab:comm_complexity}
\centering
\resizebox{\textwidth}{!}{
\begin{tabular}{>{\centering\arraybackslash}m{1.3cm} > {\centering\arraybackslash}m{2.6cm} >{\centering\arraybackslash}m{0.5cm} >{\centering\arraybackslash}m{2cm} >{\centering\arraybackslash}m{0.5cm} >{\centering\arraybackslash}m{2cm} >
{\centering\arraybackslash}m{0.5cm} >{\centering\arraybackslash}m{2cm} >{\centering\arraybackslash}m{0.5cm} >{\centering\arraybackslash}m{2cm} >{\centering\arraybackslash}m{0.5cm} >{\centering\arraybackslash}m{2cm} >{\centering\arraybackslash}m{0.5cm} >{\centering\arraybackslash}m{2cm}}
\toprule
& \multirow{ 2}{*}{Layer Component} &\multicolumn{2}{c}{Delphi Offline}&\multicolumn{2}{c}{Delphi Online}&\multicolumn{2}{c}{SecureNN}&\multicolumn{2}{c}{Falcon}&\multicolumn{2}{c}{\pmdi Offline}&\multicolumn{2}{c}{\pmdi Online}\\ 
\cmidrule{3-14}
 & & Rds. & Comm. & Rds. & Comm. & Rds. & Comm. & Rds. & Comm. & Rds. & Comm. & Rds. & Comm. \\ 
\midrule
\multirow{2}{*}{CS-C} & linear & $2$ & $2N_{enc}$ & $2$ & $2N$ & - & - & - & - & $2$ & $2N_{enc}$ & - & - \\
& non-linear & $3$ & $|GC| + 2N\kappa$ & 1 & $N \kappa$ & - & - & - & - & - & - & - & - \\
\cmidrule{1-14}
\multirow{2}{*}{C-ES} & linear & - & - & - & - & - & - & - & - & - & - & - & - \\ 
& non-linear & - & - & - & - & - & - & - & - & $6$ & $2N(2\kappa + 1)$ & - & - \\ 
\cmidrule{1-14}
\multirow{2}{*}{ES-ES} & linear & - & - & - & - & $2$ & $(h^2g^2i + 2g^2oi + h^2o)N$ & $1$ & $(h^2o)N$ & - & - & $2$ & $(h^2i + h^2o)N$ \\
& non-linear & - & - & - & - & $10$ & $(8\log p + 24)N$ & $5 + \log N$ & $0.5N$ & $1$ & $|GC|$ & $1$ & $N \kappa$ \\
\bottomrule
\end{tabular}}
\end{table*}

\subsection{\label{sec:commOver}Communication Overhead} We analyze the communication overhead of \pmdi as compared to Delphi \cite{delphi}, SecureNN \cite{securenn}, and Falcon \cite{wagh2020falcon} as summarized in  Table \ref{tab:comm_complexity}. In this context, $T$ represents the number of colluding edge servers in a cluster, $N$ is the number of bits in a single input data $x \in \mathbb{Z}_q$, $\kappa$ denotes the length of the garbled circuit labels, $N_{enc}$ is the number of bits in homomorphically encrypted data, and $|GC|$ indicates the number of bits required to transmit the garbled circuits. Additionally, $h$ denotes the dimension of the square input matrix to a layer, $g$ is the dimension of the kernel corresponding to the layer, and $i$ and $o$ represent the number of input and output channels, respectively. $p$ is the smaller field size in SecureNN \cite{securenn}.
 
Each row in Table \ref{tab:comm_complexity} shows the number of communication rounds and the amount of data (bits) exchanged between the following pairs: the cloud server and the client, the client and the edge server, and two edge servers. We exclude the communication overhead between the cloud server and edge servers, as they are usually connected via high-speed links.
The critical insight from this table is that \pmdi offloads as much communication as possible to the edge servers in the offline phase. In contrast to Delphi, which follows a client-server model, \pmdi shifts the bottleneck (sending the GCs) to the edge servers and eliminates online phase communication between the client and the cloud server. SecureNN does not separate the protocol into online and offline phases and has high communication overhead. Only Falcon has better communication overhead as compared to \pmdi, but it is limited in the sense that it (i) only supports a few non-linear functions, (ii) has higher linear computation overhead, 
(iii) does not support model-distributed inference, and (iv) cannot tolerate collusion among multiple parties due its use of replicated secret sharing. Section \ref{sec:simulation results} confirms our analysis and shows that \pmdi reduces ML inference time as compared to Delphi, SecureNN, and Falcon.
More details are provided in Appendix C.



%% file: evaluation.tex
\section{Evaluation} \label{sec:simulation results}

In this section, we evaluate the performance of our protocol \pmdi in a real testbed. We will first describe our experimental setup and then provide our results. 

\textbf{Experimental Setup.} To run our experiments, we used a desktop computer with an Intel Core i7-8700 CPU at 3.20GHz with 16GB of RAM as the cloud server and different instances of Jetstream2 \cite{jetstream2} from ACCESS \cite{NSFACCESS} as the edge servers and the client. Specifically, for the first four clusters (of edge servers) of \pmdi, we used \texttt{m3.quad} as the computing edge server, \texttt{m3.xl} as the garbler server, and \texttt{m3.2xl} as the evaluator server. For clusters 4 to 6, we used \texttt{m3.quad} as the computing edge server, \texttt{m3.large} as the garbler server, and \texttt{m3.xl} as the evaluator server. For clusters 7 and 8, we used \texttt{m3.quad} as the computing edge server, \texttt{m3.medium} as the garbler server, and \texttt{m3.large} as the evaluator server. We used an instance of \texttt{m3.medium} as the client. All the devices are connected via TCP connections. We compare \pmdi with baselines Delphi \cite{delphi}, SecureNN \cite{securenn}, and Falcon \cite{wagh2020falcon}. 

To ensure a fair comparison with Delphi, we utilized the same instance types for the client (\texttt{m3.medium} on Jetstream2) and the cloud server (a desktop computer equipped with an Intel Core i7-8700 CPU at 3.20 GHz with 16 GB of RAM) as described in the \pmdi setup.
Additionally, we conducted experiments with various Jetstream2 instance types to verify that Delphi's computational setup does not negatively impact their latency. For our comparison with Falcon, we employed three \texttt{m3.2xl} instances on Jetstream2.

We used the following ML architectures and datasets to run the experiments: 
(i) The ML model in Figure 13 of MiniONN \cite{oblivious} on the CIFAR-10 dataset; (ii) ResNet32 model introduced in \cite{he2016deep} with CIFAR-100 dataset; (iii) The ML model in Figure 12 of MiniONN \cite{oblivious} on MNIST dataset, with Maxpool being replaced by Meanpool; and (iv) VGG16 model introduced in \cite{simonyan2014very} on Tiny ImageNet dataset.

\begin{figure}[h]
\centering
\subfigure[Offline. CIFAR-10 on MiniONN]{ \scalebox{0.42}{\includegraphics[width=0.52\textwidth]{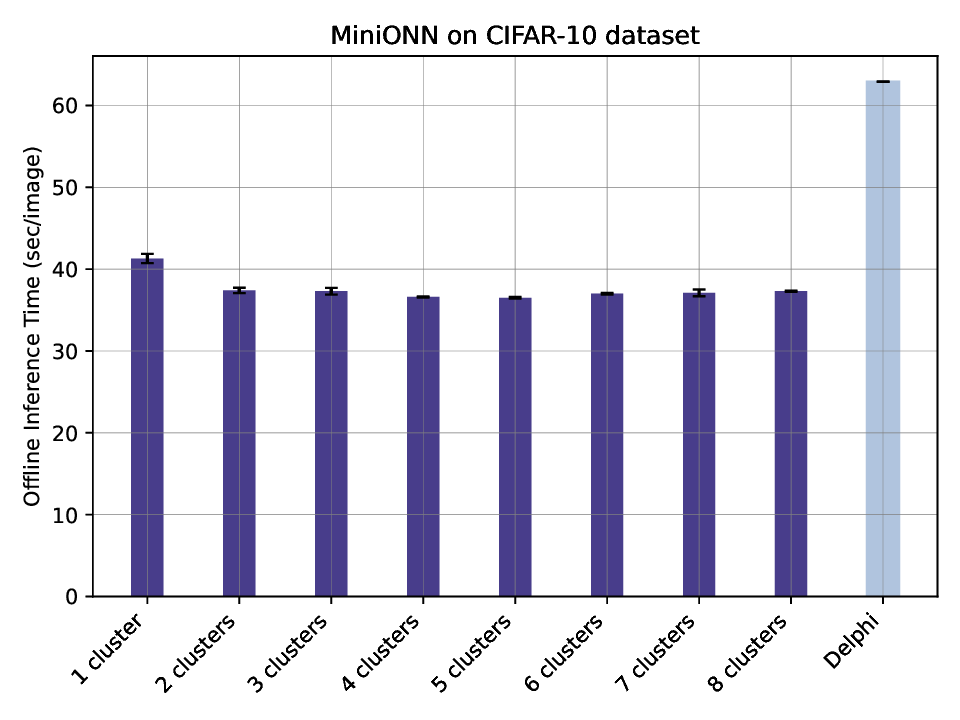}} } 
\subfigure[Online. CIFAR-10 on MiniONN]{ \scalebox{0.42}{\includegraphics[width=0.52\textwidth]{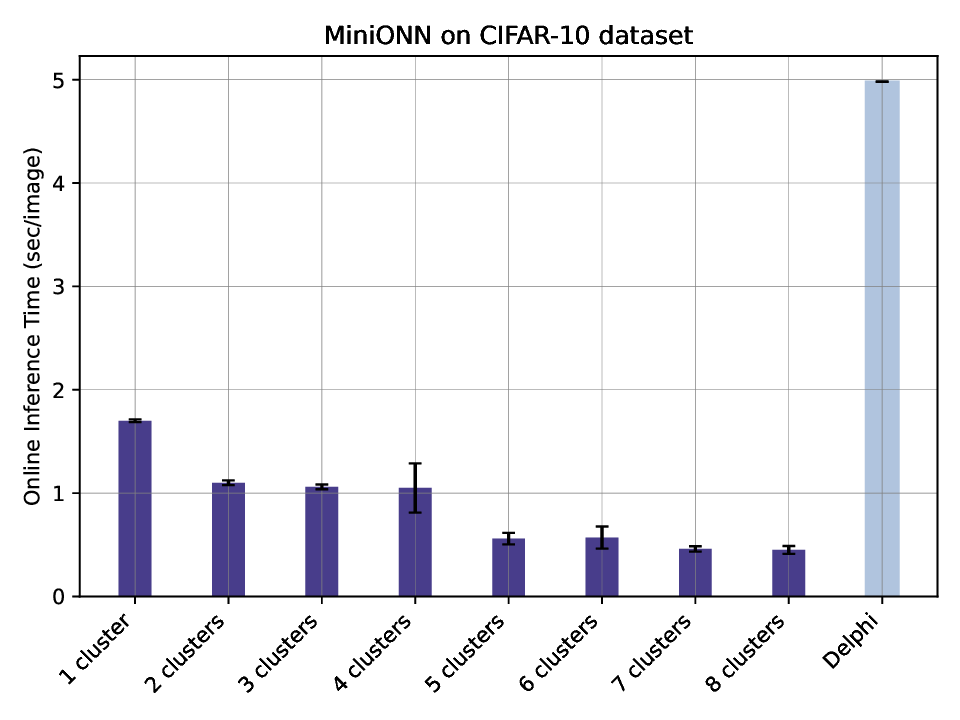}} } \\
\subfigure[Offline. CIFAR-100 on ResNet32]{ \scalebox{0.42}{\includegraphics[width=0.52\textwidth]{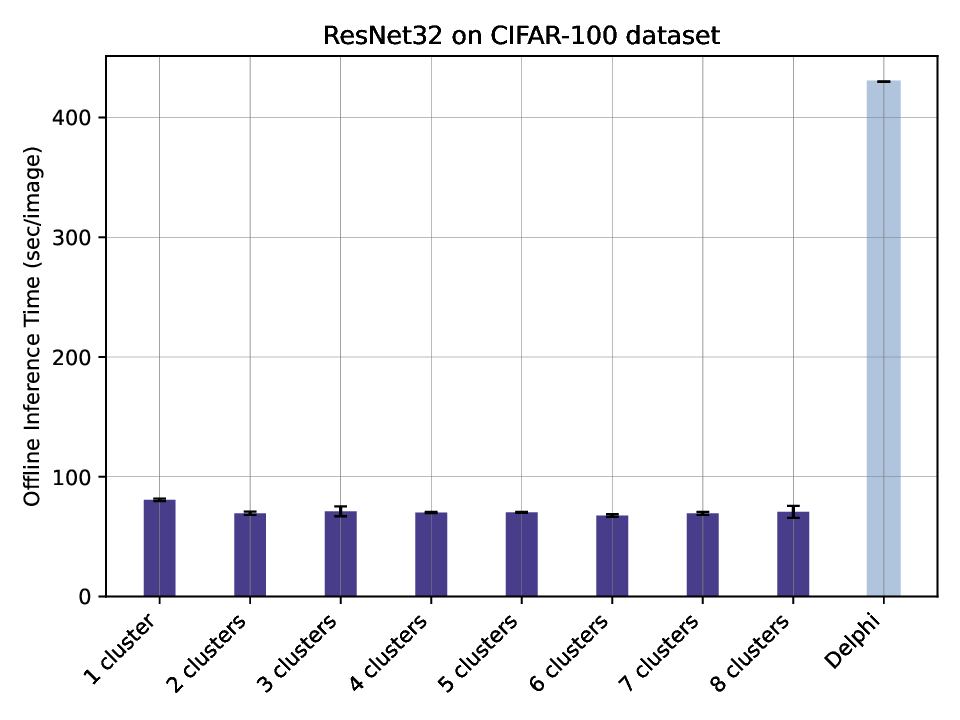}} }
\subfigure[Online. CIFAR-100 on ResNet32]{ \scalebox{0.42}{\includegraphics[width=0.52\textwidth]{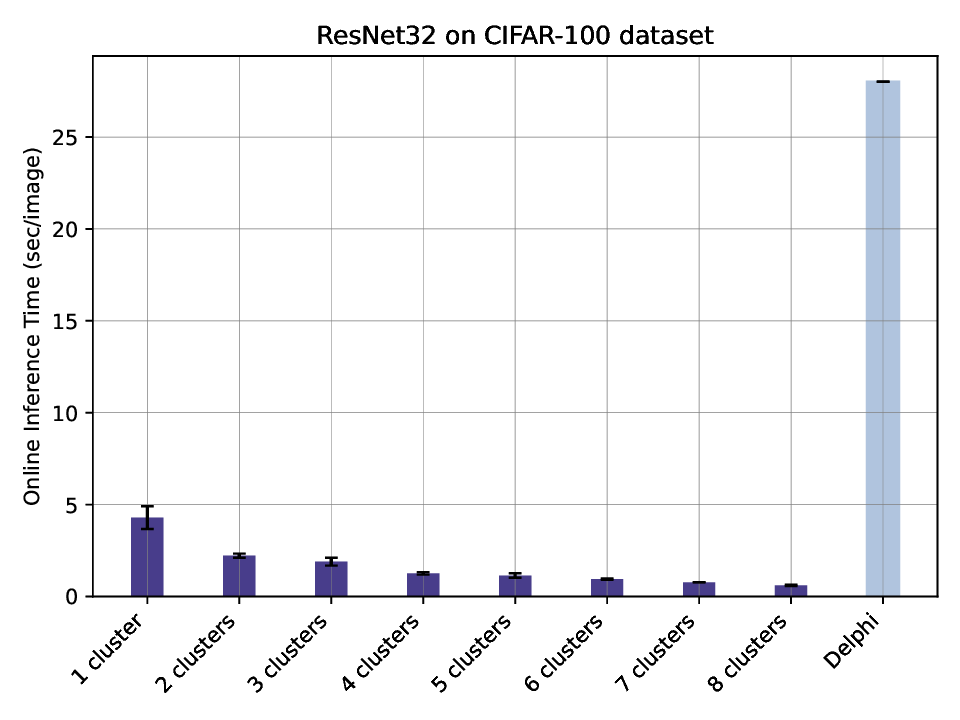}} }
\caption{ML inference time of \pmdi in online and offline phases with increasing number of clusters.}
\label{fig:pmdi-vs-delphi}
\end{figure}

\textbf{Results.} Fig. \ref{fig:pmdi-vs-delphi} presents the ML inference time of \pmdi as compared to Delphi in both offline and online phases. Fig. \ref{fig:pmdi-vs-delphi}(a) shows the delay of the offline protocols for both \pmdi and Delphi for the CIFAR-10 dataset on MiniONN architecture. As seen, \pmdi significantly improves the offline delay as compared to Delphi, thanks to reducing the communication overhead. As seen, the increasing number of clusters does not affect the offline delay in \pmdi, because model-distributed inference only affects the delay in the online phase. Indeed, Fig. \ref{fig:pmdi-vs-delphi}(b) shows that the ML inference time of \pmdi is significantly lower than Delphi in the same setup and decreases with the increasing number of clusters thanks to model distribution. 

Fig. \ref{fig:pmdi-vs-delphi}(c) and (d) show the ML inference time in the offline and online phases for the CIFAR-100 dataset and ResNet32 architecture. As seen, the improvement of \pmdi as compared to Delphi is more pronounced in this setup as the dataset and the ML model are larger in this setup, and \pmdi performs better in this setup, thanks to model parallelization and reducing communication overhead. 




\begin{table}[h]
  \centering
  \caption{ML inference time of \pmdi in the online phase as compared to Falcon and SecureNN.}
  \label{tab:sim_res}
  \resizebox{0.35\textwidth}{!}{
    \begin{tabular}{cccc}
      \toprule
      & MiniONN & AlexNet & VGG16 \\
      \midrule
     \pmdi & 0.09 s & 2.73 s & 3.63 s \\
      Falcon & 0.02 s & 2.11 s & 5.12 s \\
      SecureNN & 0.44 s & - & - \\
      \bottomrule
    \end{tabular}}
\end{table}

Table \ref{tab:sim_res} shows the ML inference time of \pmdi in the online phase as compared to SecureNN and Falcon for MiniONN, AlexNet, and VGG16 architectures. MNIST dataset is used with the MiniONN, and Tiny ImageNet dataset is used with the AlexNet and VGG16. \pmdi has 8 clusters for AlexNet and VGG16 and 7 clusters for MiniONN. 
As seen, \pmdi improves over SecureNN in MiniONN. However, Falcon's improved communication cost, as discussed in Section \ref{sec:commOver}, gives it an edge over smaller ML models like MiniONN. As the models grow larger, the performance gap narrows. While \pmdi shows improvements but does not outperform Falcon in AlexNet, it demonstrates significant improvement in the VGG16 setup due to the benefits of model-distributed inference, which become more pronounced with larger ML models and datasets. Further results of \pmdi on the VGG16 ML model for different numbers of clusters are provided in Appendix D.

%% file: conclusion.tex
\section{Conclusion} \label{sec:conc}

This paper designed privacy-preserving hierarchical model-distributed inference, \pmdi protocol to speed up ML inference in a hierarchical setup while providing privacy to both data and ML model. Our \pmdi design (i) uses model-distributed inference at the edge servers, (ii) reduces the amount of communication to/from the cloud server to reduce ML inference time, and (iii) uses additive secret sharing with HE, which reduces the number of computations.
The experimental results demonstrated that \pmdi significantly reduced the ML inference time as compared to the baselines. 



%% file: proxy_ot_proof.tex
\section*{Appendix A: Correctness and Security Proof of PROXY-OT}

In this section, we prove the correctness and security of the PROXY-OT protocol described in Section \ref{sec:system} and summarized in Fig. \ref{fig:proxy-ot}. First, we prove that the protocol is correct, meaning that at the end of the protocol, the evaluator server gets $\mathbf{k}_i$, such that $i \in \{0, 1\}$ is the private input of the client and $\mathbf{k}_0$ and $\mathbf{k}_0$ are the labels generated by the garbler server, with $\mathbf{k}_j \in \{0, 1\}^\kappa, j = 0, 1$. As stated before, $b \in \{0, 1\}$ and $\mathbf{u}_j \in \{0, 1\}^\kappa, j = 0, 1$ are random samples generated using a pseudorandom generator with the same seed. Note that the arithmetic operations $\oplus$ are done over the field $\mathbb{Z}_2$.

\emph{Correctness:} The evaluator server decrypts the masked label $\mathbf{k}_i \oplus \mathbf{u}_{b}$ knowing $\mathbf{u}_0$, $\mathbf{u}_1$, and $b$:
    \begin{align}
        (\mathbf{k}_i \oplus \mathbf{u}_{b}) \oplus \mathbf{u}_b = \mathbf{k}_i
        \nonumber
    \end{align}

\emph{Security:} At each round, the exchanged messages are masked by a random sample. Knowing this, it can be concluded that each ciphertext reveals no information. In particular: 

\begin{itemize}
\item $i \oplus b$ is secure, according to Shannon's perfect secrecy theorem:
    \begin{align}
        \forall i \in \{0, 1\}, i' \in \{0, 1\}: \nonumber \\
        \Prob(i \oplus b = i') = \frac{1}{2}
        \nonumber
    \end{align}

\item $\mathbf{k}_0 \oplus \mathbf{u}_{i \oplus b}$ and $\mathbf{k}_1 \oplus \mathbf{u}_{\overline{i \oplus b}}$ are secure, according to Shannon's perfect secrecy theorem:
    \begin{align}
        \forall \mathbf{k}_j \in \{0, 1\}^\kappa, \mathbf{k}' \in \{0, 1\}^\kappa: \nonumber \\ 
        \Prob(\mathbf{k}_0 \oplus \mathbf{u}_{i \oplus b} = \mathbf{k}') = \Prob(\mathbf{k}_1 \oplus \mathbf{u}_{\overline{{i \oplus b}}} = \mathbf{k}') = \left(\frac{1}{2}\right)^\kappa
        \nonumber
    \end{align}

\item $\mathbf{k}_i \oplus \mathbf{u}_b$ can be proved secure similarly.
    \begin{align}
        \forall i \ \mathbf{k}_i \in \{0, 1\}^\kappa, \mathbf{k}' \in \{0, 1\}^\kappa: \nonumber \\ 
         \Prob(\mathbf{k}_i \oplus \mathbf{u}_b = \mathbf{k}') = \left(\frac{1}{2}\right)^\kappa
        \nonumber
    \end{align}

\end{itemize}

%% file: protocol_proof.tex
\section*{Appendix B: Security Proof of \pmdi} \label{pmdi_sec_proof}

To demonstrate security for \pmdi, for each adversarial scenario, a simulator is presented, followed by a proof of indistinguishability with the real-world distribution. We use hybrid arguments to prove the indistinguishability of the real-world distribution and the ideal-world (simulation) distribution for each case. In the simulator's description, we focus on the actions taken by the simulator for clarity and simplicity. 


\noindent \textbf{Compromised cloud server.}

\noindent \emph{Simulator:} The simulator $\mathcal{S}_{\textnormal{CS}}$, when provided with the cloud server’s input $\mathbf{M}_1, \ldots , \mathbf{M}_L$, proceeds as follows during the offline phase. 

\begin{enumerate}
    \item $\mathcal{S}_{\textnormal{CS}}$ chooses a uniform random tape for the cloud server.
    
    \item $\mathcal{S}_{\textnormal{CS}}$ chooses a public key \texttt{pk} for a linearly homomorphic encryption scheme.
                
    \item For the first layer, $\mathcal{S}_{\textnormal{CS}}$ sends \texttt{HE.Enc(pk,$\mathbf{r}'_1$)} instead of \texttt{HE.Enc(pk,$\mathbf{r}_1$)} to the cloud server, where $\mathbf{r}'_1$ is chosen randomly from $\mathbb{Z}^n_q$.
                
    \item For every layer $l$ that has a non-linear component, $\mathcal{S}_{\textnormal{CS}}$ sends a junk ciphertext \texttt{HE.Enc(pk,$\mathbf{r}'_l$)} instead of \texttt{HE.Enc(pk,$\mathbf{r}_l$)} to the cloud server, where $\mathbf{r}'_l$ is chosen randomly from $\mathbb{Z}^n_q$.
\end{enumerate}

\noindent \emph{Indistinguishability proof:}

    $\pmb{\mathcal{H}}_{0}$\textbf{:} This hybrid corresponds to the real world distribution where the actual input data $\mathbf{x}$ is used as the client's input.
    
    $\pmb{\mathcal{H}}_{1}$\textbf{:} We make a syntactic change in this hybrid. In the offline phase, for the layers with a non-linear component and the first layer, the client sends \texttt{HE.Enc(pk,$\mathbf{r}'_l$)} to the cloud server instead of encrypted $\mathbf{r}_l$. From the semantic security of the homomorphic encryption, it can be concluded that $\pmb{\mathcal{H}}_{0}$ and $\pmb{\mathcal{H}}_{1}$ are computationally indistinguishable. Also, $\pmb{\mathcal{H}}_{1}$ is identically distributed to the simulator's output, which completes the proof for this part.


\break
\noindent \textbf{Compromised client.}

\noindent \emph{Simulator:} The simulator $\mathcal{S}_{\textnormal{C}}$, when provided with the client’s input $\mathbf{x}$, proceeds as follows.
    \begin{enumerate}
        \item $\mathcal{S}_{\textnormal{C}}$ chooses a uniform random tape for the client. 
        
        \item In the offline phase:
        \begin{enumerate}
            \item For every layer $l$ with a non-linear component:
            \begin{enumerate}
                \item $\mathcal{S}_{\textnormal{C}}$ sends \texttt{HE.Enc(pk,$\mathbf{s}'_l$)} as the homomorphically evaluated ciphertext 
                to the client, $\mathbf{s}'_l$ is chosen randomly from $\mathbb{Z}^n_q$.
                
                \item $\mathcal{S}_{\textnormal{C}}$ uses $\mathcal{S}_{\textnormal{GS}}$ to generate a random circuit with a random output on $1^\lambda$ and $1^{|C|}$ as inputs. 
                
                \item During the PROXY-OT protocol, for each input bit $i$ of the client, the simulator $\mathcal{S}_{\textnormal{C}}$ sends two random samples $\mathbf{k}'_0$ and $\mathbf{k}'_1$ as the masked labels 
                to the client.
            \end{enumerate}
            \item If $l$ is the last layer $L$, $\mathcal{S}_{\textnormal{C}}$ sends \texttt{HE.Enc(pk,$\mathbf{s}'_l$)} as the homomorphically evaluated ciphertext 
            to the client.
        \end{enumerate}

        \item In the online phase, $\mathcal{S}_{\textnormal{C}}$ sends $\mathbf{y} - \mathbf{c}_{L+1}$ to the client as the output of the layer evaluations.
            
    \end{enumerate}

\noindent \emph{Indistinguishability proof.}

    $\pmb{\mathcal{H}}_{0}$\textbf{:} This hybrid corresponds to the real-world distribution where the actual model parameters $\mathbf{M}_1 , \ldots , \mathbf{M}_L$ are used as the cloud server's input.
    
    $\pmb{\mathcal{H}}_{1}$\textbf{:} In the offline phase, we use the simulator for the function privacy of \texttt{HE}, $\mathcal{S}_{\texttt{FP}}$. So, for every homomorphic evaluation, $\mathcal{S}_{\texttt{FP}}$ uses $\mathbf{M}_l \cdot \mathbf{c}_l + \mathbf{s}_l$ as input to generate the evaluated ciphertext. From the function privacy of the \texttt{HE}, it can be concluded that $\pmb{\mathcal{H}}_{1}$ and $\pmb{\mathcal{H}}_{2}$ are computationally indistinguishable.
    
    $\pmb{\mathcal{H}}_{2}$\textbf{:} In this hybrid, we change the input of $\mathcal{S}_{\texttt{FP}}$ to be a random sample $\mathbf{s}'_l$ from $\mathbb{Z}^n_q$. This hybrid is identically distributed as the previous hybrid because of the semantic security of \texttt{HE}.

    $\pmb{\mathcal{H}}_{3}$\textbf{:} For the PROXY-OT protocol, the garbler server sends two random labels $\mathbf{k}'_0$ and $\mathbf{k}'_1$ to the client, instead of the actual masked labels. From the information-theoretic security of the PROXY-OT protocol, this hybrid is indistinguishable from the previous hybrid.
    
    $\pmb{\mathcal{H}}_{4}$\textbf{:} In the last step, the simulator $\mathcal{S}_C$ sends $\mathbf{y} - \mathbf{c}_{L+1}$ to the client, where $\mathbf{y} = \mathbf{M}(\mathbf{x})$ is the simulator's input. Since this is a syntactic change, $\pmb{\mathcal{H}}_{3}$ and $\pmb{\mathcal{H}}_{4}$ are computationally indistinguishable. Note that this hybrid is also identically distributed to the simulator's output, completing the proof. 


\noindent \textbf{Compromised computing edge server.}

\noindent \emph{Simulator:} The simulator $\mathcal{S}_{\textnormal{ES},C}$ proceeds as follows.
    \begin{enumerate}
        \item $\mathcal{S}_{\textnormal{ES},C}$ chooses a uniform random tape for the compromised computing server.
        
        \item In the offline phase, for every layer $l$, $\mathcal{S}_{\textnormal{ES},C}$ sends random $[\mathbf{M}'_l]_v$ and $[\mathbf{s}'_l]_v$ chosen from $\mathbb{Z}_q^n$ as model parameter shares and cloud server's randomness to the compromised computing server.

        \item In the online phase, for every layer $l$, $\mathcal{S}_{\textnormal{ES},C}$ sends a uniformly random sample $\mathbf{c}'_l$ to the compromised computing server as the input to the layer that it should process.
    \end{enumerate} 

\noindent \emph{Indistinguishability proof:}

        $\pmb{\mathcal{H}}_{0}$\textbf{:} This hybrid corresponds to the real-world distribution where the actual client's input $\mathbf{x}$ and cloud server's model parameters $\mathbf{M}_1, \ldots , \mathbf{M}_L$ are used.

        $\pmb{\mathcal{H}}_{1}$\textbf{:} In the offline phase, the cloud server sends $[\mathbf{M}'_l]_v$ and $[\mathbf{s}'_l]_v$ instead of $[\mathbf{M}_l]_v$ and $[\mathbf{s}_l]_v$ to the computing edge server. Since $[\mathbf{M}'_l]_v$ and $[\mathbf{s}'_l]_v$ are uniformly random, this hybrid is indistinguishable from the previous one.
        
        $\pmb{\mathcal{H}}_{2}$\textbf{:} In the online phase, for every layer $l$, the garbler server sends a random sample $\mathbf{c}'_l$ from $\mathbb{Z}^n_q$ to the computing edge server instead of the actual layer input $\mathbf{x}_l - \mathbf{c}_l$. As $\mathbf{c}'_l$ is indistinguishable from the actual input, it can be concluded that $\pmb{\mathcal{H}}_{1}$ and $\pmb{\mathcal{H}}_{2}$ are indistinguishable. Also, $\pmb{\mathcal{H}}_{2}$ is identically distributed to the simulator's output, which completes the proof for this part.


\noindent \textbf{Compromised garbler server.} 

\noindent \emph{Simulator:} The simulator $\mathcal{S}_{\textnormal{ES},G}$ proceeds as follows.
    \begin{enumerate}
        \item $\mathcal{S}_{\textnormal{ES},G}$ chooses a uniform random tape for the compromised garbler server.
        \item In the offline phase:
        \begin{enumerate}
            \item For every layer $\mathcal{S}_{\textnormal{ES},G}$ sends junk shares, $[\mathbf{M}'_l]_v$ as model parameter shares and $[\mathbf{s}'_l]_v$ as cloud server's randomness to the compromised garbler server.

            \item To run the PROXY-OT protocol, $\mathcal{S}_{\textnormal{ES},G}$ sends a random bit $b'$ as the masked input of the client to the garbler server.
        \end{enumerate}
        
        \item In the online phase:
        \begin{enumerate}
            \item $\mathcal{S}_{\textnormal{ES},G}$ sends a uniformly random sample $\mathbf{c}'_1$ as the first layer's input to the garbler server.
            
            \item For each layer $l$, $\mathcal{S}_{\textnormal{ES},G}$ sends $T$ random junk terms, corresponding to the outputs of the $T$ computing edge servers for that layer, to the garbler server.

            \item If layer $l$ has a non-linear component, $\mathcal{S}_{\textnormal{ES},G}$ sends a random junk share $\mathbf{c}'_{l+1}$ to the garbler server as the output of the evaluator server, which is the input of the next layer.
        \end{enumerate}
    \end{enumerate}

\noindent \emph{Indistinguishability proof:}

        $\pmb{\mathcal{H}}_{0}$\textbf{:} This hybrid corresponds to the real-world distribution where the actual client's input $\mathbf{x}$ and cloud server's model parameters $\mathbf{M}_1, \ldots , \mathbf{M}_L$ are used.

        $\pmb{\mathcal{H}}_{1}$\textbf{:} In the offline phase, the cloud server sends $[\mathbf{M}'_l]_v$ and $[\mathbf{s}'_l]_v$ instead of $[\mathbf{M}_l]_v$ and $[\mathbf{s}_l]_v$ to the garbler server. Since $[\mathbf{M}'_l]_v$ and $[\mathbf{s}'_l]_v$ are uniformly random, this hybrid is indistinguishable from the previous one.

        $\pmb{\mathcal{H}}_{2}$\textbf{:} In the offline phase, the client sends a random bit $b'$ instead of the actual input to the garbler server for every input in the PROXY-OT protocol. As $b'$ is indistinguishable from the actual input, we can conclude that $\pmb{\mathcal{H}}_{2}$ is indistinguishable from $\pmb{\mathcal{H}}_{1}$.

        $\pmb{\mathcal{H}}_{3}$\textbf{:} In the online phase, for the first layer, the client sends a random sample $\mathbf{c}'_1$ from $\mathbb{Z}^n_q$ to the garbler server instead of the actual input $\mathbf{x}_1 - \mathbf{c}_1$. As $\mathbf{c}'_1$ is indistinguishable from the actual input, it can be concluded that $\pmb{\mathcal{H}}_{2}$ and $\pmb{\mathcal{H}}_{3}$ are indistinguishable.

        $\pmb{\mathcal{H}}_{4}$\textbf{:} In the online phase, for every layer $l$, the $T$ computing edge servers send $[\mathbf{s}'_l]_v$ instead of $[\mathbf{M}_l]_v \cdot \mathbf{c}_l + [\mathbf{s}_l]_v$ to the garbler server. As $[\mathbf{s}'_l]_v$ is indistinguishable from the actual output, it can be concluded that $\pmb{\mathcal{H}}_{3}$ and $\pmb{\mathcal{H}}_{4}$ are indistinguishable.

        $\pmb{\mathcal{H}}_{5}$\textbf{:} In the online phase, for every layer $l$ with a non-linear component, the evaluator server sends a random sample $\mathbf{c}'_{l+1}$ from $\mathbb{Z}^n_q$ to the garbler server instead of the actual output $\mathbf{x}_{l+1} - \mathbf{c}_{l+1}$. As $\mathbf{c}'_{l+1}$ is indistinguishable from the actual output, it can be concluded that $\pmb{\mathcal{H}}_{4}$ and $\pmb{\mathcal{H}}_{5}$ are indistinguishable. We complete the proof by noting that this hybrid is identically distributed to the simulator's output.


\noindent \textbf{Compromised evaluator server.} 

\noindent \emph{Simulator:} The simulator $\mathcal{S}_{\textnormal{ES},E}$ proceeds as follows.
    \begin{enumerate}
        \item $\mathcal{S}_{\textnormal{ES},E}$ chooses a uniform random tape for the compromised evaluator server.
        
        \item In the offline phase:
            \begin{enumerate}
                \item For every layer $l$ with a non-linear component, $\mathcal{S}_{\textnormal{ES},E}$ uses $\mathcal{S}_{\textnormal{GS}}$ and runs it on $1^\lambda$, $1^{|C|}$ and sets the output of the circuit to be a random value. $\mathcal{S}_{\textnormal{GS}}$ sends $\Tilde{C}$ to the compromised evaluator server. 
                    
                \item To run the PROXY-OT protocol, $\mathcal{S}_{\textnormal{ES},E}$ randomly sends one sets of the output labels generated by $\mathcal{S}_{\textnormal{GS}}$ to the evaluator server. Noting that the exchanged labels correspond to the inputs of the GC generated in the offline phase.
            \end{enumerate}
            
            \item In the online phase, for every layer $l$ with a non-linear component, $\mathcal{S}_{\textnormal{ES},E}$ sends randomly one set of the labels generated by $\mathcal{S}_{\textnormal{GS}}$ to the evaluator server. Noting that the exchanged labels correspond to the input of the GC generated in the online phase.
                   
        \end{enumerate}

\noindent \emph{Indistinguishability proof:}

        $\pmb{\mathcal{H}}_{0}$\textbf{:} This hybrid corresponds to the real-world distribution where the actual client's input $\mathbf{x}$ and cloud server's model parameters $\mathbf{M}_1, \ldots , \mathbf{M}_L$ are used.

        $\pmb{\mathcal{H}}_{1}$\textbf{:} For the execution of PROXY-OT, the client sends a random label $\mathbf{k}_{b'}$ to the evaluator server where $b'$ is a random bit in executing PROXY-OT. As $b'$ is indistinguishable from the actual input $i \oplus b$, this hybrid is identically distributed to the previous one.

        $\pmb{\mathcal{H}}_{2}$\textbf{:} In the online phase, for every layer $l$ with a non-linear component, the garbler server sends $\mathbf{k}_{b'}$ to the evaluator server where $b'$ is a random bit. As all the inputs of the GC are secret shares, $b'$ is identically distributed to the actual input. Hence hybrid $\pmb{\mathcal{H}}_{2}$ is identically distributed to hybrid $\pmb{\mathcal{H}}_{1}$.

        $\pmb{\mathcal{H}}_{3}$\textbf{:} In the offline phase, for every layer $l$ with a non-linear component, we generate $\Tilde{C}$ using $\mathcal{S}_{\text{GS}}$ on input $1^\lambda$, $1^{|C|}$ and $C(\mathbf{w},\mathbf{z})$. As $C(\mathbf{w},\mathbf{z})$ is an additive secret share of $\mathbf{x}_{l+1}$, this hybrid is identically distributed to the previous one. Finally, we note that $\pmb{\mathcal{H}}_{3}$ is identically distributed to the simulator's output; hence, the proof is complete.

%% file: analysis_extended.tex
\section*{Appendix C: More on Communication Overhead}

As stated previously, Table \ref{tab:comm_complexity} provides a detailed comparison of the communication overhead between \pmdi and other frameworks such as Delphi, SecureNN, and Falcon. For the \emph{linear} rows, we assume that the layer in question is a convolutional layer with dimensions $(o, g, g)$ and an input dimension of $(i, h, h)$. Consequently, the communication overhead for a fully connected layer can be easily derived. For the \emph{non-linear} layer component, we assumed the activation function to be ReLU. We note that the communication overheads for Falcon and SecureNN are directly taken from their respective papers, while we independently calculated the overhead for Delphi. Table \ref{tab:notation} summarized the notation we use for the communication overhead analysis. 

\begin{table}[h]
    \caption{Notations used for the communication overhead analysis.}
    \label{tab:notation}
    \centering
    \begin{tabular}{>{\centering\arraybackslash}m{1.5cm} >{\arraybackslash}m{6cm}}
        \toprule
        Symbol & Description \\
        \midrule
        $N_{enc}$ & Number of bits of homomorphically encrypted data. \\
        $N$ & Number of bits of a single data $x \in \mathbb{Z}_q$. \\
        $|GC|$ & Number of bits to send the garbled circuit and corresponding inner labels. \\
        $\kappa$ & Length of the labels of GC. \\
        $h$ & Layer input dimension. \\
        $g$ & Dimension of the kernel of a layer. \\
        $i$ & Number of input channels. \\
        $o$ & Number of output channels. \\
        $p$ & Smaller field size in SecureNN \cite{securenn}. \\
        \bottomrule
    \end{tabular}
\end{table}

\pmdi emphasizes reducing online latency. By incorporating an offline phase, we halve the number of multiplications and minimize the communication complexity for convolutional layers by a factor proportional to the kernel size. These optimizations lead to lower latency compared to SecureNN. 
Furthermore, the computational complexity of \pmdi remains constant regardless of the number of computing nodes. In contrast, Falcon's computational complexity increases proportionally with the number of computing nodes on each server.
We note that existing work SECO \cite{chen2024seco} introduces additional homomorphic encryption computations for the second split (partition) of the ML model and increases the homomorphically encrypted communication overhead by a factor of three in the offline phase. Furthermore, during the online phase, the first part of the neural network requires an interactive protocol with the client, similar to Delphi, which introduces communication and computation cost. 



%% file: evaluation_extended.tex
\begin{figure}[ht]
\centering
\subfigure[Offline. TI on VGG16]{ \scalebox{0.42}{\includegraphics[width=0.52\textwidth]{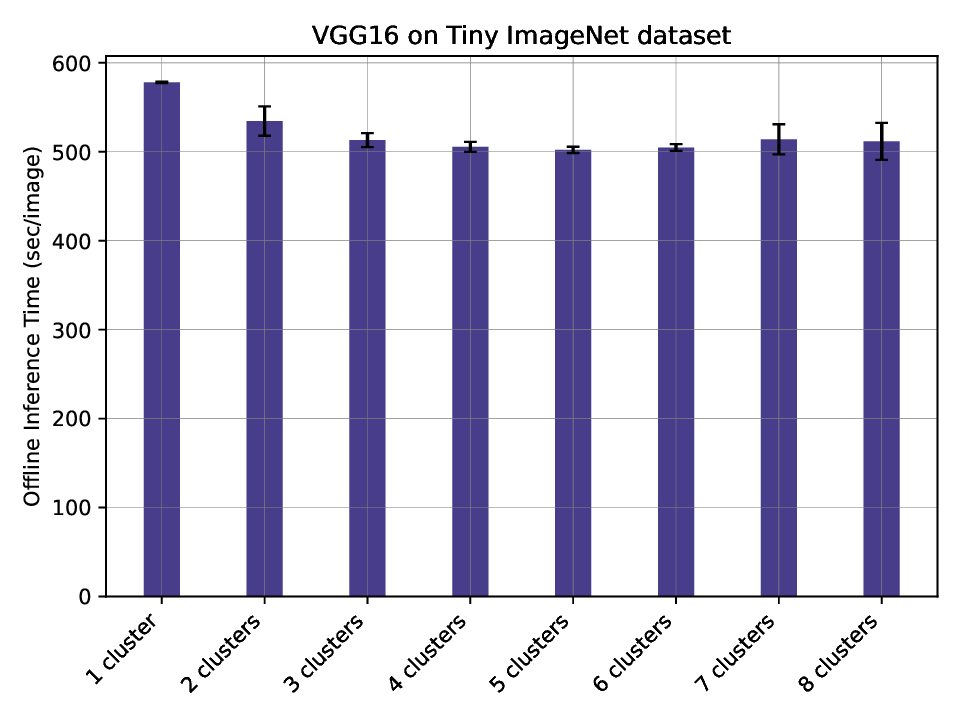}} }
\subfigure[Online. TI on VGG16]{ \scalebox{0.42}{\includegraphics[width=0.52\textwidth]{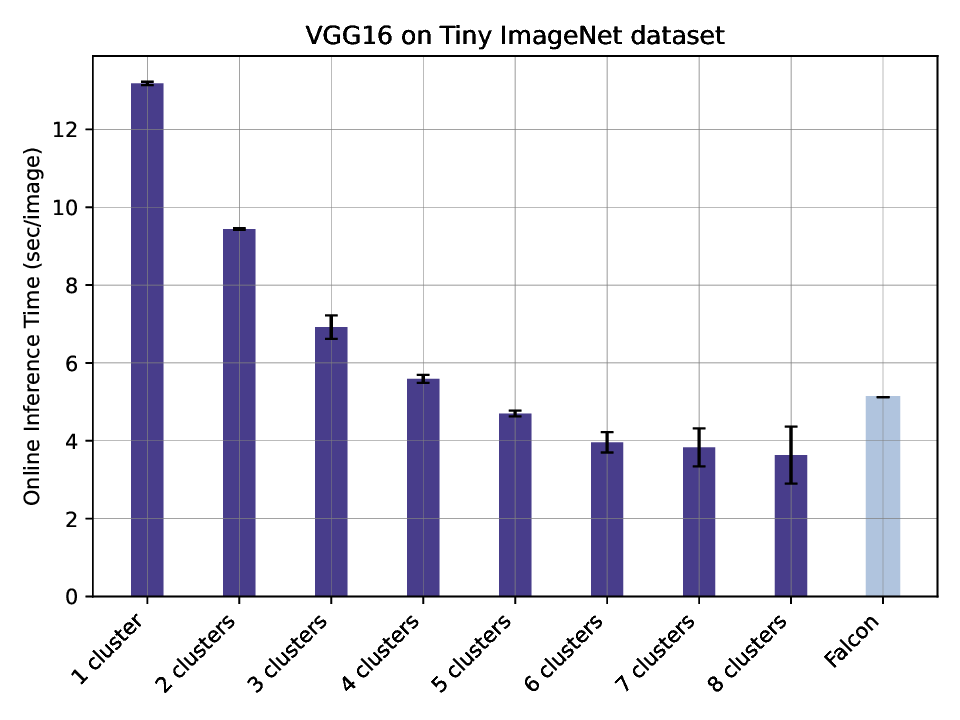}} }
\caption{ML inference time of \pmdi in online and offline phases with increasing number of clusters.}
\label{fig:pmdi-vs-falcon}
\end{figure}

\section*{Appendix D: Extended Evaluation}

Figure \ref{fig:pmdi-vs-falcon} presents a detailed evaluation of the VGG16 ML model's performance with varying cluster numbers in both offline and online phases. Fig. \ref{fig:pmdi-vs-falcon}(a) shows the delay of the offline protocol of \pmdi for the Tiny ImageNet (TI) dataset on VGG16 architecture, while Fig. \ref{fig:pmdi-vs-falcon}(b) shows the delay of the online protocol of \pmdi versus Falcon. As mentioned earlier, the main advantage of \pmdi over Falcon lies in the distribution of ML model layers. With an increasing number of clusters, \pmdi's performance becomes significantly better as compared to Falcon in the online phase.